\documentclass[10pt,onecolumn,twoside]{IEEEtran} 

\IEEEoverridecommandlockouts                              

\usepackage{amsmath,amsfonts,amssymb}
\usepackage{algorithm}
\usepackage{algpseudocode}
\usepackage{enumitem}
\usepackage{array}
\usepackage[noadjust]{cite}
\usepackage{textcomp}
\usepackage{verbatim}
\usepackage{graphicx}
\usepackage{tikz}
\usepackage{multirow}
\usepackage{pgfplots}
\usepackage{grffile}
\pgfplotsset{compat=newest}
\usepgfplotslibrary{groupplots}   
\usetikzlibrary{plotmarks}
\usetikzlibrary{arrows.meta}
\usepgfplotslibrary{patchplots}
\usepackage{siunitx}
\usetikzlibrary{shapes,arrows}
\tikzstyle{block} = [draw, fill=white, rectangle, minimum height=3em, minimum width=4em]
\tikzstyle{sum} = [draw, fill=white, circle, node distance=1cm]
\usetikzlibrary{backgrounds}
\usetikzlibrary{arrows.meta,calc,angles,quotes,decorations.pathmorphing}

\usepackage{hyperref}
\usepackage{cleveref}
\crefname{assmpt1}{Assumption}{Assumptions}
\Crefname{assmpt1}{Assumption}{Assumptions}

\crefrangeformat{assmpt1}{Assumptions~#3#1#4--#5#2#6}
\Crefrangeformat{assmpt1}{Assumptions~#3#1#4--#5#2#6}

\newcommand{\bbR}{\mathbb{R}}
\newcommand{\bbN}{\mathbb{N}}
\newcommand{\bbZ}{\mathbb{Z}}

\newcommand{\calD}{\mathcal{D}}
\newcommand{\calA}{\mathcal{A}}
\newcommand{\calB}{\mathcal{B}}
\newcommand{\calZ}{\mathcal{Z}}
\newcommand{\calS}{\mathcal{S}}
\newcommand{\calR}{\mathcal{R}}
\newcommand{\calH}{\mathcal{H}}
\newcommand{\calX}{\mathcal{X}}
\newcommand{\calK}{\mathcal{K}}

\newcommand{\calC}{\mathcal{C}}

\newcommand{\calI}{\mathcal{I}}
\newcommand{\calN}{\mathcal{N}}
\newcommand{\calU}{\mathcal{U}}
\newcommand{\calY}{\mathcal{Y}}

\newcommand{\rmd}{\mathrm{d}}
\newcommand{\rmr}{\mathrm{r}}

\newcommand{\bfK}{\mathbf{K}}

\newcommand{\bfk}{\mathbf{k}}
\newcommand{\bfu}{\mathbf{u}}
\newcommand{\bfy}{\mathbf{y}}
\newcommand{\bfzero}{\mathbf{0}}

\newcommand{\sfto}{\mathsf{to}}
\newcommand{\sfin}{\mathsf{in}}

\newcommand{\RMSE}{\mathsf{RMSE}}

\newcommand{\Sim}{\mathsf{Sim}}

\newcommand{\ra}{\rightarrow}

\DeclareMathOperator*{\argmin}{arg\,min}

\usepackage{amsthm}

\newtheorem{thm1}{\bf Theorem}
\newtheorem{prop1}{\bf Proposition}
\newtheorem{lem1}{\bf Lemma}
\newtheorem{assmpt1}{\bf Assumption}
\newtheorem{defn1}{\bf Definition}
\newtheorem{rem1}{\bf Remark}

\newtheorem{cor1}{\bf Corollary}

\newtheorem{prob1}{\bf Problem}

\newenvironment{defn}{\begin{defn1}}{\hfill$\square$\end{defn1}}
\newenvironment{asm}{\begin{assmpt1}}{\hfill$\square$\end{assmpt1}}
\newenvironment{rem}{\begin{rem1}}{\hfill$\square$\end{rem1}}
\newenvironment{lem}{\begin{lem1}}{\hfill$\square$\end{lem1}}
\newenvironment{thm}{\begin{thm1}}{\hfill$\square$\end{thm1}}

\newenvironment{prop}{\begin{prop1}}{\hfill$\square$\end{prop1}}

\title{\LARGE \bf
Inverse Learning-Based Output Feedback Control of Nonlinear Systems with Verifiable Guarantees
}

\author{Yeongjun Jang, Hamin Chang, Heein Park, Hyeonyeong Jang, Takashi Tanaka, Hyungbo Shim
\thanks{*This work was supported by the National Research Foundation of Korea(NRF) grant funded by the Korea government(MSIT) (No. RS-2022-00165417) and the grant from Hyundai Motor Company's R\&D Division.
}
\thanks{Y.~Jang, H.~Park, H.~Jang, and H.~Shim are with ASRI, Department of Electrical and Computer Engineering, Seoul National University, Seoul, 08826, Korea (email: jangyj0512@snu.ac.kr, \{phoinee,hyjang\}@cdsl.kr).
}
\thanks{H.~Chang, and T.~Tanaka are with the Elmore Family School of Electrical and Computer Engineering, Purdue University, West Lafayette, IN, 47907, USA (email: \{chan1232,tanaka16\}@purdue.edu).
}
\thanks{T.~Tanaka is also with the School of Aeronautics and Astronautics, Purdue University, West Lafayette, IN, 47907, USA.
}
}

\begin{document}

\maketitle
\thispagestyle{plain} 
\pagestyle{plain} 

\begin{abstract}
In this paper, we present a data-driven output feedback controller for nonlinear systems that achieves practical output regulation, using noise-free input/output measurement data.
The proposed controller is based on (i) an inverse model of the system identified via kernel interpolation, which maps a desired output and the current state to the corresponding desired control input; and (ii) a data-driven reference selection framework that actively chooses a suitable desired output from the dataset which has been used for the identification. 
We establish a verifiable sufficient condition on the dataset under which the proposed controller guarantees practical output regulation.
Numerical simulations demonstrate the effectiveness of the proposed controller, with additional evaluations in the presence of output measurement noise to assess its robustness empirically.
\end{abstract}

\section{Introduction}\label{sec:intro}

Data-driven control aims to design controllers directly from measured data, thereby avoiding the need to identify an explicit mathematical model of the underlying system.
This method has emerged as an effective alternative to model-based control, particularly when obtaining a model is too costly or requires substantial domain expertise \cite{HouzWang13}. 
In the case of linear systems, for example, data-driven control has been extensively studied \cite{DepeTesi19,CoulLyge19,VanwEisi20,BerbAllg25,VanwCaml25}, building on the behavioral approach and the fundamental lemma introduced in \cite{WillRapi05}.

Recently, there has been growing interest in data-driven control for nonlinear systems.
The main challenge in this direction is to establish theoretical closed-loop guarantees \cite{MartScho23}, and existing works that provide such guarantees rely on conditions that are often difficult to verify in practice.
For example, \cite{BerbKohl22} proposed a model predictive control (MPC) scheme that exploits local linearizations of the system, whose closed-loop guarantees require recursive feasibility of the MPC optimization problem. 
The works of \cite{DaitSzna21} and \cite{GuomDepe21} formulated data-dependent optimization problems based on linear matrix inequality (LMI) and sum-of-squares (SOS) conditions, respectively, to synthesize stabilizing controllers.  
In \cite{StraBerb23}, the Koopman operator is employed to lift a nonlinear system into a higher-dimensional linear system, enabling an LMI-based control design that ensures stability. 
Although these results provide formal guarantees, verifying the feasibility of the associated MPC optimization problem, and the LMI and SOS conditions can be nontrivial and computationally intensive.

In this context, kernel methods have become a popular tool for data-driven control of nonlinear systems due to their ability to provide practical and reliable regression error bounds \cite{SrinKrau12,ChowGopa17,FiedSche21Err,ReedLaur25,SchaMadd22}.
Moreover, their nonparametric nature allows them to represent a broad class of functions \cite{SchoSmol02,WillRasm06,KanaHenn18}.
These methods have mainly been utilized to identify a \textit{forward model} of the system, followed by the design of a controller that compensates for residual (learning) errors that may occur \cite{Koci16}.
However, the resulting models are often highly nonlinear with limited physical meaning, which complicates the controller design.
Consequently, most existing approaches have been restricted to MPC frameworks \cite{HewiKabz19,MaddScha21KPC,NguyPfef22,HuanLyge23,DeJoLaza24}, which require solving an optimization problem at each time step. 
This incurs substantial computational burden, and theoretical guarantees are typically stated under the assumption that the online optimization problem is recursively feasible, which is often difficult to verify.
Other existing approaches have focused on designing robust linear controllers under certain simplifications, such as assuming bounded nonlinearities \cite{FiedSche21} or considering control-affine systems with constant input matrices \cite{HuzhDepe23}.

On the other hand, there have been approaches that utilize kernel methods to identify an \textit{inverse model} of the system \cite{WillKlan08,RelaMuno23,KimhChan23,TanaFagi17,JangChan24}, which maps a desired output and the current state to the corresponding desired control input.
This obviates the need for designing a controller for the identified model because it can directly act as a tracking controller.
Provided that the reference trajectory is feasible, \cite{KimhChan23} and \cite{TanaFagi17} established closed-loop guarantees on tracking performance.
However, determining whether a reference trajectory is feasible without explicit knowledge of the system dynamics is generally unrealistic, which limits the practicality of such guarantees.

To address this limitation, \cite{JangChan24} introduced a data-driven framework that actively selects suitable reference points from the input/state data used to identify the inverse model.
The key idea is to leverage interpolation/regression error bounds for kernel methods to construct, for each data point, a region of the state space where using that point as a reference ensures a desirable closed-loop behavior.  
Building upon this framework, \cite{JangChan24} established a verifiable sufficient condition on the dataset under which a data-driven state feedback controller renders the closed-loop system ultimately bounded.

In this paper, we propose a data-driven \textit{output} feedback controller for nonlinear systems that achieves practical output regulation.
The controller is designed based on an inverse model of the system identified from noise-free input/output measurements via kernel interpolation, and thus, does not require full state measurements.
Based on the error bound of kernel interpolation, we derive a verifiable sufficient condition on the dataset under which the proposed controller guarantees practical output regulation.
The effectiveness of the proposed controller is demonstrated through numerical simulations, and further tests carried out in the presence of output measurement noise indicate that the proposed controller can remain effective under noisy environments.

Specifically, we focus on systems represented in the nonlinear autoregressive exogenous (NARX) form, a popular framework for modeling nonlinear dynamics \cite{CadeRive16,Bill13,MishMark21,PisoFari09}.
NARX models can be reformulated into an augmented state-space representation, where a one-step ahead output is fully determined by an augmented state comprising past inputs and outputs.
A direct application of the framework in \cite{JangChan24}, however, would ultimately bound the entire augmented state, which unnecessarily drives both input and output sequences toward zero. 
This is particularly problematic for systems that require time-varying or oscillatory control inputs to regulate the output.
In this regard, we adopt and generalize the reference selection framework of \cite{JangChan24} to accommodate the augmented state formulation.

The remainder of the paper is organized as follows. 
Section~\ref{sec:prob} defines the inverse model of a system and formulates the problem.
Section~\ref{sec:ilc} presents a method to identify the inverse model using kernel interpolation and analyze its error bound. 
Section~\ref{sec:control} provides explicit construction of the proposed data-driven controller, and Section~\ref{sec:sim} demonstrates simulation results. 
Finally, Section~\ref{sec:conclusion} concludes the paper.

\emph{Notation:} 
Let $\bbR$, $\bbR_{\ge 0}$, $\bbZ$, $\bbZ_{\ge 0}$, and $\bbN$ denote the set of real numbers, nonnegative real numbers, integers, nonnegative integers, and positive integers, respectively. 
For $r\in \bbR_{\ge 0}$ and $x\in \bbR^n$, a closed ball of radius $r$ centered at $x$ is denoted by $\calB(x,r):=\{ z\in \bbR^n \mid \|x-z \|\le r \}$.
The set of all $x\in \bbR$ such that $a\le x\le b$ is denoted by $[a,b]$.
For a sequence $v_1,\ldots,v_n$ of vectors or scalars, we define $\begin{bmatrix}
  v_1;\cdots;v_n  
\end{bmatrix}:= 
\begin{bmatrix}
    v_1^\top \cdots v_n^\top    
\end{bmatrix}^\top$.
For a signal $v:\bbZ \ra \bbR^n$ and integers $a,b\in\bbZ$ such that $a\le b$, we define $v_{[a,b]} := 
\begin{bmatrix}
    v(a); v(a+1) ; \cdots ; v(b)    
\end{bmatrix}$.
Let $\bfzero_{m\times n}\in\bbZ^{m\times n}$ and $I_n\in\bbZ^{n\times n}$ denote the zero matrix and identity matrix, respectively. 
We use the shorthand notation $\{x_i\}_{i=1}^n:=\{x_1,x_2,\ldots,x_n\}$ to denote the set of elements $x_i$ for $i=1,\ldots,n$.

\section{Problem Formulation}\label{sec:prob}

Consider a discrete-time system represented by a nonlinear autoregressive exogenous (NARX) model of the form
\begin{align}\label{eq:sysNARX}
        y(t+1) =f(y_{[t-n+1,t]}, u_{[t-n+1,t]}),
\end{align}
where $u(t)\in\bbR^m$ is the input, $y(t)\in\bbR^p$ is the output, and $n\in\bbN$ is the model order.
For the sake of simplicity, we focus on the case $m=p=1$ and demonstrate that our result can be extended to the multi-input multi-output case in Remark~\ref{rem:mimo}.
It is assumed that the function $f$ is unknown, while $n$ is known.

The goal of this paper is to design a data-driven output feedback controller using input/output measurement data of \eqref{eq:sysNARX} that achieves practical output regulation with a desired accuracy $\delta>0$ in finite time.
That is, for some $\kappa\in \bbN$ that depends on the accuracy $\delta$ and the initial condition $(y_{[-n+1,0]},u_{[-n+1,0]})$ at $t=0$, the controller should ensure 
\begin{equation}\label{eq:goal}
    \| y(t) \|\le \delta \quad \forall t\ge \kappa.
\end{equation}
In addition, we aim to establish a verifiable sufficient condition on the dataset under which the proposed controller guarantees \eqref{eq:goal}. 

In line with \cite{KimhChan23} and \cite{JangChan24}, we identify an inverse model of \eqref{eq:sysNARX} and utilize it as a foundation of the controller.
In order to formally define an inverse model, the following definitions and assumption are introduced. 
Let us rewrite system \eqref{eq:sysNARX}, with a slight abuse of notation, as 
\begin{equation}\label{eq:sys}
    y(t+1) = f(\zeta(t), u(t)),
\end{equation}
where
\begin{equation}\label{eq:zeta}
    \zeta(t) := \begin{bmatrix}
             y_{[t-n+1,t]} \\ u_{[t-n+1,t-1]}
         \end{bmatrix} \in \bbR^{2n-1}
\end{equation}
is referred to as the \textit{augmented state} at time step $t\in \bbZ$, consisting of the most recent $n$ outputs and $n-1$ inputs. 
The set of feasible augmented states is defined as 
\begin{align*}\label{eq:feasible}
    \calZ:=\{\zeta\in\bbR^{2n-1} \mid \exists u,y:\bbZ\to\bbR~
    \mbox{satisfying} ~\eqref{eq:sys}~ 
    \mbox{for all}~t\in\bbZ~ \mbox{and}~\zeta=\zeta(\tau)~\mbox{for some}~\tau\in\bbZ  \}.
\end{align*}
In other words, $\calZ\subset \bbR^{2n-1}$ collects all augmented states that are realizable by some input/output trajectory of system \eqref{eq:sys}.
For any $\zeta\in\calZ$, we define the one-step reachable set of outputs by 
\begin{equation}\label{eq:reachable}
    \calR(\zeta) := \left\{y^+ \in \bbR \mid \exists u\in\bbR \ \mbox{such that} \ y^+=f(\zeta,u)  \right\},
\end{equation}
and assume that the mapping $f(\zeta,\cdot)$ is injective for all $\zeta\in\calZ$ as follows.
\begin{asm}\label{asm:inverse}\upshape
    For any $\zeta\in\calZ$ and $y^+\in\calR(\zeta)$, there exists a unique $u\in\bbR$ such that $y^+ = f(\zeta,u)$.
\end{asm}
Assumption~\ref{asm:inverse} implies that system \eqref{eq:sys} has a global relative degree one, since the current input $u(t)$ directly influences the next output $y(t+1)$ and is uniquely determined from $\zeta(t)$ and $y(t+1)$. 
In Section~\ref{subsec:extend}, we discuss how our framework can be extended to NARX models with input delays, which corresponds to the case with global relative degree greater than one.

Under Assumption~\ref{asm:inverse}, we define the inverse model of \eqref{eq:sys} as follows.

\begin{defn}\label{def:inverse}\upshape
    The inverse model of system \eqref{eq:sys} is defined as a function $c:\bbR^{2n}\ra \bbR$ such that
    \begin{equation}\label{eq:inv}
            y^+ = f(\zeta,c([y^+;\zeta]))
    \end{equation}
    for all $\zeta\in\calZ$ and $y^+\in\calR(\zeta)$.
\end{defn}

The inverse model $c$ can be interpreted as a function that maps a desired output $y^+$ and an augmented state $\zeta$ to the corresponding desired control input $c([y^+;\zeta])$ such that \eqref{eq:inv} holds.
Definition~\ref{def:inverse} is implicit, as it does not specify the value of $c([y^+;\zeta])$ when $\zeta\notin \calZ$ or $y^+\notin \calR(\zeta)$.  
However, we emphasize that any function $c$ satisfying Definition~\ref{def:inverse} can serve as the inverse model in our approach, and its specific choice does not affect the theoretical formulation or results.
For example, one may choose the inverse model as the function $c$ that renders $f(\zeta,c([y^+;\zeta]))$ to be the projection of $y^+$ onto $\calR(\zeta)$, as in \cite{JangChan24}.

\begin{rem}\upshape
    Suppose that system \eqref{eq:sys} is affine-in-control, written as 
    \begin{equation}\label{eq:sysEx}
        y(t+1) = h(\zeta(t))+g(\zeta(t))\cdot u(t)   
    \end{equation}
    with some functions $h:\bbR^{2n-1}\to \bbR$ and $g:\bbR^{2n-1}\to \bbR$, where $g(\zeta)\ne 0$ for all $\zeta\in\calZ$. 
    Then, it follows that $\calR(\zeta)=\bbR$ for all $\zeta\in\calZ$, and Assumption~\ref{asm:inverse} holds since the input $u\in\bbR$ satisfying $y^+=f(\zeta,u)$ is uniquely determined by $u=(y^+-h(\zeta))/g(\zeta)$.
\end{rem}

The inverse model $c$ can be directly utilized as a feedback controller that tracks a given reference trajectory;
given $y_\rmr(t+1)\in\calR(\zeta(t))$, a controller of the form 
\begin{equation}\label{eq:invcontrol}
    u(t) = c([y_\rmr(t+1);\zeta(t)])
\end{equation}
achieves 
\begin{equation*}
    \begin{split}
        y(t+1) = f(\zeta(t),c([y_\rmr(t+1);\zeta(t)])) = y_\rmr(t+1)
    \end{split}
\end{equation*}
by Definition~\ref{def:inverse}.
However, the design of controller \eqref{eq:invcontrol} requires both exact knowledge of the inverse model $c$ and a reference point $y_\rmr(t+1)$ that is one-step reachable from $\zeta(t)$, clearly necessitating knowledge of the function $f$.

In this context, we aim to propose an inverse learning-based controller of the form 
\begin{equation}\label{eq:igp}
    u(t) = \hat{c}([y_\rmr(t+1);\zeta(t)])
\end{equation}
that guarantees \eqref{eq:goal}, where (i) $\hat{c}$ is a data-driven identification of the inverse model $c$ based on input/output measurement data of \eqref{eq:sys}; 
(ii) the reference point $y_\rmr(t+1)$ is actively chosen from the dataset used for the identification.

Since $f$ is assumed to be unknown, we cannot, in general, verify one-step reachability of a chosen reference point. 
Thus, we make a technical assumption that possibly conservative Lipschitz constants of the function $f$ and its inverse model $c$ are known.
This will enable us to quantify the deviation of the actual output $y(t+1)$, obtained by applying \eqref{eq:igp}, from the reference point $y_\rmr(t+1)$.

  \begin{asm}\label{asm:lipschitz}\upshape
     There exist known constants $L_f>0$ and $L_c>0$ such that 
    \begin{align*}
        \left\| f(\zeta,u) - f(\zeta',u') \right\| \le L_f \left\| 
        \begin{bmatrix}
            \zeta-\zeta' \\
            u - u'
        \end{bmatrix} 
        \right\|, ~~~~~~
        \left\| c(\xi) - c(\xi') \right\| \le L_c  \| \xi-\xi' \| 
    \end{align*}
    for all $\zeta,\zeta'\in\bbR^{2n-1}$, $u,u'\in \bbR$, and $\xi,\xi'\in\bbR^{2n}$.
\end{asm}

\begin{rem}\upshape\label{rem:uniformObsv}
    We show that nonlinear systems of the form
    \begin{align} \label{eq:sysState}
        x(t+1) &= F(x(t),u(t)), \\
        y(t) &= h(x(t)), \nonumber
    \end{align}
    can be equivalently reformulated into the NARX form \eqref{eq:sys} if it is uniformly $n$-observable \cite{MoraGriz02}.
    Here, $x(t)\in\bbR^{n_x}$ is the state, and $u(t)\in\bbR$ and $y(t)\in\bbR$ are the input and output, respectively, as in \eqref{eq:sys}. 
    Uniform $n$-observability is frequently assumed in nonlinear observer design \cite{Hanb09,GautHamm02}, as it enables the reconstruction of the initial state from a finite number of input/output measurements.   
    That is, it implies that there exists a mapping $H:\bbR^n\times \bbR^n \to \bbR^{n_x}$ such that 
    \begin{equation*}
        x(t-n+1) =H(y_{[t-n+1,t]}, u_{[t-n+1,t]})
    \end{equation*}
    for any pair of input sequence $u_{[t-n+1,t]}$ and output sequence $y_{[t-n+1,t]}$ of system \eqref{eq:sysState}, and the corresponding state $x(t-n+1)$.
    
    Now, let $\Phi:\bbR^{n_x}\times \bbR^n \to \bbR^{n_x}$ denote the $n$-step state transition mapping that satisfies
    \begin{equation*}
        x(t+1) = \Phi(x(t-n+1),u_{[t-n+1,t]}).
    \end{equation*}
    Substituting this into \eqref{eq:sysState} yields
    \begin{align*}
        y(t+1) &= h(x(t+1))\\
        &= h(\Phi(x(t-n+1),u_{[t-n+1,t]})) \\
        &= h(\Phi(H(y_{[t-n+1,t]}, u_{[t-n+1,t]}),u_{[t-n+1,t]})) =:f(\zeta(t),u(t)),
    \end{align*}
    which is in the NARX form \eqref{eq:sys}.
\end{rem}

\section{Inverse Model Identification and
Error Bounds}\label{sec:ilc}

\subsection{Learning the inverse model}\label{subsec:inv}

For the design of controller \eqref{eq:igp}, we employ kernel interpolation (KI) to identify the inverse model $c$ as $\hat{c}$.
 To conduct the identification, we first collect and rearrange input/output data of system \eqref{eq:sys} into input/output data of the inverse model $c$. 
Consider a $T$-length input/output trajectory\footnote{The superscript $\rmd$ indicates sample data collected from offline experiments conducted prior to control design.} of system \eqref{eq:sys} denoted by $(u_{[0,T-1]}^\rmd,y_{[0,T]}^\rmd)$, where $T\ge n$.
Then,
\begin{equation*}
    y^\rmd(i+n-1) = f(\zeta^\rmd(i+n-2),u^\rmd(i+n-2))
\end{equation*}
for all $1\le i \le T-n+1$, where
\begin{equation*}
    \zeta^\rmd(i+n-2) := 
    \begin{bmatrix}
        y_{[i-1,i+n-2]}^\rmd \\ u_{[i-1,i+n-3]}^\rmd  
    \end{bmatrix}\in \bbR^{2n-1}.  
\end{equation*}
By the definition of the inverse model $c$,
\begin{equation*}
    u^\rmd(i+n-2) = c([y^\rmd(i+n-1);\zeta^\rmd(i+n-2)])
\end{equation*}
holds for all $1\le i \le T-n+1$, and this allows us to construct a training dataset $\cal{D}$ that consists of input/output data of the inverse model $c$: 
\begin{equation}\label{eq:dataset}
    \calD:=\{([y_i^+;\zeta_i],u_i) \}_{i=1}^N,
\end{equation}
where $N:=T-n+1$ is the number of training data and 
\begin{equation*}
    [y_i^+;\zeta_i]:=[y^\rmd(i+n-1);\zeta^\rmd(i+n-2)], ~~ u_i := u^\rmd(i+n-2),
\end{equation*}
are the training input and output, respectively.
For notational simplicity, we also define
\begin{equation*}
    \begin{split}
        \xi_i:=[y_i^+;\zeta_i], ~~~~
        \calD_\sfin := \{ \xi_i\}_{i=1}^N, ~~~~ \calD_\sfto := \{ y_i^+ \}_{i=1}^N.
    \end{split}
\end{equation*}

It is emphasized that the training dataset $\calD$ can alternatively be constructed from multiple short input/output trajectories of \eqref{eq:sys}. 
This aspect is especially advantageous when the system is unstable, as it allows the collection of a sufficiently large training dataset while avoiding numerical instability issues.

Given the dataset $\calD$, KI obtains an estimate of $c$ by appropriately interpolating the data points.
To understand KI, two key concepts are crucial: (strictly) positive definite kernels, which measure similarity between data points, and reproducing kernel Hilbert space (RKHS), which provides a mathematical framework for deriving interpolation error bounds.

\begin{defn}\label{def:pdk}\upshape
    For a nonempty set $\calX$, a symmetric function $k:\calX\times \calX \ra \bbR$ is called a kernel. 
    The kernel $k$ is positive definite if for any $N\in\bbN$ and any $\xi_1,\ldots,\xi_N\in \calX$, the matrix
    \begin{equation}\label{eq:Gram}
        \bfK:=\begin{bmatrix}
            k(\xi_1,\xi_1) & \cdots & k(\xi_1,\xi_N) \\
            \vdots & \ddots & \vdots \\
            k(\xi_N,\xi_1) & \cdots & k(\xi_N,\xi_N)
        \end{bmatrix}
    \end{equation}
    is positive semidefinite. 
    The kernel $k$ is strictly positive definite if $\bfK$ is positive definite for any $N\in\bbN$ and any pairwise distinct $\xi_1,\ldots,\xi_N\in \calX$.
\end{defn}
\begin{defn}\label{def:RKHS}\upshape
    For a positive definite kernel $k:\calX \times \calX \ra \bbR$, a Hilbert space $\calH$ of real-valued functions $h:\calX\ra\bbR$ equipped with an inner product $\langle \cdot, \cdot \rangle_\calH$ is said to be an RKHS with the reproducing kernel $k$ if
    \begin{equation*}
        k(\cdot,\xi)\in\calH \ \, \mbox{and} \ \, h(\xi)=\langle h, k(\cdot,\xi) \rangle_\calH
    \end{equation*}
    for all $\xi \in \calX$ and $h\in\calH$.
\end{defn}

By the Moore-Aronszajn theorem \cite{AronNach50}, there exists a one-to-one correspondence between RKHSs and positive definite kernels.
In this regard, we denote by $\calH_k$ the unique RKHS having a positive definite kernel $k$ as the reproducing kernel.
Also, we denote the norm induced by the inner product of $\calH_k$ by $\| \cdot \|_{\calH_k}:=\langle \cdot,\cdot \rangle_{\calH_k}$.

\begin{rem}\upshape
    The abstract definition of RKHS given in Definition~\ref{def:RKHS} is further explored.
    For a positive definite kernel $k:\calX \times \calX \ra \bbR$, the function $k(\cdot,\xi)$ is called the canonical feature map of $\xi\in\calX$. 
    The RKHS $\calH_k$ can then be explicitly constructed as the closure of the set of finite linear combinations of canonical feature maps \cite[Section~2.3]{KanaHenn18}:
    \begin{equation}\label{eq:RKHSdef}
        \calH_k= \left\{h=\sum_{i=1}^\infty \alpha_i k(\cdot,\xi_i) \mid \ \alpha_1,\alpha_2,\ldots \in \bbR, \ \xi_1,\xi_2,\ldots \in \calX, ~\|h\|_{\calH_k}^2 := \sum_{i=1}^\infty \sum_{j=1}^\infty \alpha_i \alpha_j k(\xi_i,\xi_j) < \infty \right\}.
    \end{equation}
    This construction shows that the elements of $\calH_k$ naturally inherit geometric properties of $k$, such as smoothness or periodicity.
\end{rem}


Let us fix $k:\bbR^{2n}\times\bbR^{2n}\ra\bbR$ as a  strictly positive definite kernel for the remainder of this paper. 
Given the dataset $\calD$ in \eqref{eq:dataset}, KI seeks the minimal norm interpolant in $\calH_k$ by solving the following optimization problem:
\begin{equation}\label{eq:krr}
    \hat{c} = \argmin_{\mathsf{c}\in\calH_k} \left\| \mathsf{c} \right\|_{\calH_k}^2 \ \ \mbox{s.t.} \ \ \mathsf{c}(\xi_i)=u_i \ \ \forall i=1,\ldots,N.
\end{equation}
Thanks to the celebrated representer theorem \cite[Theorem~16.1]{Wend04}, the closed-form solution of the infinite-dimensional optimization problem \eqref{eq:krr} can be written as
\begin{equation}\label{eq:chatCF}
    \hat{c}(\xi)=\bfk^\top(\xi)\bfK^{-1}\bfu,
\end{equation}
where $\bfK\in\bbR^{N\times N}$ is the matrix in \eqref{eq:Gram} constructed from $\calD_{\sfin}$, and
$\bfk:\bbR^{2n} \to \bbR^N$ and $\bfu\in\bbR^N$ are defined by
\begin{align*}
    \bfk(\xi) := 
    \begin{bmatrix}
        k(\xi_1,\xi) ;
        \cdots ;
        k(\xi_N,\xi)
    \end{bmatrix}, ~~~~~
    \bfu := \begin{bmatrix}
        u_1 ; \cdots ; u_N
    \end{bmatrix}.
\end{align*}
By removing redundant data from the training dataset, we can readily ensure that the training inputs $\xi_1,\ldots,\xi_N\in\calD_\sfin$ are pairwise distinct. 
Since the kernel $k$ is strictly positive definite, this guarantees the invertibility of the matrix $\bfK$.

\subsection{Error bounds}\label{subsec:err}
A notable feature of KI is that it provides an explicit upper bound on the error between the true inverse model $c$ and the estimate $\hat{c}$ defined by \eqref{eq:chatCF}.
This feature will play a central role in establishing formal guarantees for the proposed controller. 
In fact, other learning methods could be employed as long as they provide an explicit error bound as KI.

In order to derive a practical and reliable error bound, we impose a regularity condition on the complexity of the inverse model $c$.

\begin{asm}\label{asm:RKHSnorm}\upshape
    The inverse model $c$ of system \eqref{eq:sys} belongs to $\calH_k$, and $\| c\|_{\calH_k}\le \Gamma$ for some known constant $\Gamma\ge 0$. 
\end{asm}

The RKHS norm quantifies the smoothness of a function with respect to the associated reproducing kernel; the larger the norm, the less smooth the function is \cite[Section~2.3]{KanaHenn18}. 
Obtaining a reliable yet non-conservative upper bound on the RKHS norm of a target function remains an open problem \cite[Section~6.1]{FiedMenn24}.
Nevertheless, the RKHS norm has been a central tool in deriving error bounds for KI and Gaussian process regression \cite{ChowGopa17,FiedSche21Err,MaddScha21}, and is often assumed to be available in related works due to its theoretical utility \cite{FiedSche21,NguyPfef22,HuzhDepe23}.
In practice, an upper bound $\Gamma$ can be obtained from data using a safe Bayesian optimization algorithm in \cite{TokmKris24}, or by scaling $\|\hat{c}\|_{\calH_k}=\bfu^\top \bfK^{-1}\bfu$ with a safety factor greater than $1$, since $\|\hat{c}\|_{\calH_k}\le \|c\|_{\calH_k}$ always holds \cite[Section~16]{Wend04}.

An intuition behind the derivation of an error bound for $\hat{c}$ is that the estimate is expected to be more accurate at test points that are closer to a training point.
In this regard, we focus on the case in which the kernel $k$ is isotropic and decreasing, which naturally provides stronger correlations to data points that are closer to each other.

\begin{defn}\label{def:iso}\upshape
    A kernel $k:\bbR^{2n}\times\bbR^{2n} \to \bbR$ is isotropic and decreasing if there exists a decreasing function $\bar{k}:\bbR_{\ge 0} \to \bbR$ such that 
    \begin{equation}\label{eq:iso}
        k(\xi,\xi')=\bar{k}(\|\xi-\xi'\|) 
    \end{equation}
    for all $\xi,\xi'\in\bbR^{2n}$.
\end{defn}

Considering a strictly positive definite, isotropic, and decreasing kernel is not restrictive, as it includes many commonly used kernels, such as squared-exponential, Laplacian, Mat\'ern, and rational quadratic kernels.
The following lemma presents an error bound for $\hat{c}$ under this setting.

\begin{lem}[{\cite[Lemma~1]{JangChan24}}]\label{lem:errbound}\upshape
    Consider a strictly positive definite kernel $k:\bbR^{2n}\times \bbR^{2n}\to \bbR$ that is isotropic and decreasing, and let $\hat{c}$ be the corresponding kernel interpolant given by \eqref{eq:chatCF}.
    Then, there exists a known class $\calK$ function $\eta:\bbR_{\ge 0}\ra\bbR_{\ge 0}$ such that for any $\xi \in \bbR^{2n}$, 
    \begin{equation}\label{eq:errbound}
        \|c(\xi) - \hat{c}(\xi) \| \le \eta(\epsilon(\xi)),
    \end{equation}
    where $\epsilon(\xi):=\min_{\xi_i\in\calD_\sfin} \|\xi_i-\xi \|$.
\end{lem}


Although tightening the upper bound function $\eta$ is beyond the scope of this paper, it is worth noting that utilizing a tighter $\eta$ directly enhances the performance of the proposed controller.

\section{Inverse Learning-Based Control}\label{sec:control}

In this section, we complete the design of controller \eqref{eq:igp} by presenting a data-driven framework that actively selects an appropriate reference point $y_\rmr(t+1)$ from $\calD$, based on the error bound of $\hat{c}$ stated in Lemma~\ref{lem:errbound}.
We then establish a verifiable sufficient condition on the training dataset $\calD$ under which the proposed controller guarantees \eqref{eq:goal}.

We begin by introducing an abstract formulation of the proposed framework, followed by a detailed procedure for its explicit construction and implementation.
We then demonstrate how the proposed framework can be extended to NARX models with input delays.

\subsection{Abstract formulation}

For a given $\delta>0$, let us define
\begin{equation*}
    \calS_\delta:=\{\zeta\in\bbR^{2n-1} \mid \|\zeta_n\|\le \delta\},
\end{equation*}
where $\zeta_n\in\bbR$ denotes the $n$-th component of the vector $\zeta$.
Since $y(t)$ is the $n$-th component of the augmented state $\zeta(t)$, the objective \eqref{eq:goal} is equivalently reformulated as ensuring
\begin{equation}\label{eq:goalRe}
    \zeta(t)\in\calS_\delta \quad \forall t \ge \kappa
\end{equation}
for some $\kappa\in\bbN$.

Given the dataset $\calD$, we define the \textit{backward reachable set} of an arbitrary set $\calA\subset\bbR^{2n-1}$, as
\begin{align}\label{eq:backReach}
    \calR^{-}(\calA) := \left\{\zeta(t)\in\calZ \mid \exists y_i^+\in\calD_\sfto ~ \mbox{such that}~ \eqref{eq:igp} ~\mbox{with}~ y_{\rmr}(t+1)=y_i^+ ~ \mbox{guarantees} ~ \zeta(t+1)\in \calA \right\}.
\end{align}
This set consists of $\zeta(t)$ from which controller \eqref{eq:igp}, when provided with an appropriate reference point $y_\rmr(t+1)\in\calD_\sfto$, guarantees $\zeta(t+1)\in\calA$.

Building on these definitions, we recursively define a sequence of sets $(\calA_\delta^j)_{j=0}^\infty$ that satisfies
\begin{equation}\label{eq:defProp}
    \begin{split}
        \calA_\delta^0 \subset \calS_\delta, \quad
        \calA_\delta^{j+1} \subset \calR^{-}(\calA_\delta^j) \quad \forall j\in\bbZ_{\ge 0}.   
    \end{split}
\end{equation}
Then, it can be seen that $\zeta(t)\in\calA_\delta^\kappa$ for some $\kappa\in\bbN$ is a sufficient condition under which controller \eqref{eq:igp} with appropriate reference points from the set $\calD_\sfto$ guarantees the output to be bounded by $\delta$ after $\kappa$ steps of control. 
This observation is formalized in the following lemma.

\begin{lem}\label{lem:abskstep}\upshape
    If $\zeta(t)\in\calA_\delta^\kappa$ for some $\delta>0$ and $\kappa\in\bbN$, then there exist
    $y_{i_1}^+, y_{i_2}^+, \ldots, y_{i_\kappa}^+ \in\calD_\sfto$
    such that controller $\eqref{eq:igp}$ with the reference points $y_\rmr(t+l) = y_{i_l}^+$ for $l=1,\ldots,\kappa$ guarantees 
    \begin{equation}\label{eq:kstep}
        \zeta(t+l)\in\calA_\delta^{\kappa-l} \quad \forall l=1,\ldots,\kappa,
    \end{equation}
    which implies $\|y(t+\kappa)\|\le \delta$.
\end{lem}

Note that if $\calA_\delta^0 \subset \calA_\delta^1$
holds in addition to \eqref{eq:kstep}, then we can again guarantee $\zeta(t+\kappa+1)\in\calA_\delta^0$ according to Lemma~\ref{lem:abskstep}.
Therefore, by applying the same argument repeatedly, we can constrain the output to remain bounded by $\delta$ for an infinite time horizon, thus ensuring \eqref{eq:goal}.
The following lemma captures this observation and establishes a sufficient condition under which controller \eqref{eq:igp} guarantees \eqref{eq:goal}.

\begin{lem}\label{lem:absmain}\upshape
    If $\zeta(0)\in\calA_\delta^\kappa$ and $\calA_\delta^0 \subset \calA_\delta^1$ holds for some $\delta>0$ and $\kappa\in\bbN$ then controller \eqref{eq:igp} with appropriate reference points in $\calD_\sfto$ guarantees \eqref{eq:goal}.
\end{lem}

Although selecting $\calA_\delta^0=\calS_\delta$ satisfies the defining property given in \eqref{eq:defProp}, it is highly unlikely for the inclusion $\calA_\delta^0 \subset \calA_\delta^1$ to hold in this case. 
This is because $\calS_\delta$ only constrains the $n$-th component of its element $\zeta\in\bbR^{2n-1}$ to be bounded by $\delta$, thus covering the entire space in the remaining dimensions. 
Meanwhile, choosing a smaller $\calA_\delta^0$ may cause the sets $\calA_\delta^j$ for $j\ge 1$ to shrink, making it less likely for the condition $\zeta(0)\in\calA_\delta^\kappa$ to be satisfied for some $\kappa\in\bbN$.
Hence, the set $\calA_\delta^0$ must be carefully designed considering this trade-off.

\subsection{Explicit implementation}

In what follows, we exploit the error bound of KI to explicitly construct $(\calA_\delta^j)_{j=0}^\infty$, and verify that they satisfy the defining property given in \eqref{eq:defProp}. 
Then, we provide a method to actively select appropriate reference points from the set $\calD_\sfto$, and a practical guideline to implement the proposed controller.

Consider controller \eqref{eq:igp} with the reference point chosen as $y_\rmr(t+1)=y_i^+$ for some data point $([y_i^+;\zeta_i],u_i)\in \calD$. 
If $\zeta(t)\approx \zeta_i$, it is expected that 
\begin{align*}
    u(t)=\hat{c}([y_i^+;\zeta(t)])&\approx \hat{c}([y_i^+;\zeta_i])= u_i, \\
    y(t+1)=f(\zeta(t), u(t))&\approx f(\zeta_i, u_i) =y_i^+,
\end{align*}
where $\hat{c}([y_i^+;\zeta_i])= u_i$ follows from \eqref{eq:chatCF}.
The following proposition shows that $\|u_i-u(t)\|$ and $\|y_i^+ - y(t+1)\|$ can be upper bounded by class $\calK_\infty$ functions of $\|\zeta_i - \zeta(t) \|$, which aligns with this intuition.

\begin{prop}\label{prop:bound}\upshape
    There exist known class $\mathcal{K}_\infty$ functions $\gamma_u: \bbR_{\ge 0} \ra \bbR_{\ge 0}$ and $\gamma_y: \bbR_{\ge 0} \ra \bbR_{\ge 0}$ such that for any $ ([y_i^+;\zeta_i],u_i)\in \calD$, 
    \begin{align}
        \left \| u_i - \hat{c}([y_i^+;\zeta(t)]) \right \| &\le \gamma_u(\|\zeta_i-\zeta(t) \|),  \label{eq:inputBound} \\
         \left \| y_i^+ - f(\zeta(t), \hat{c}([y_i^+;\zeta(t)])) \right \| &\le \gamma_y(\|\zeta_i-\zeta(t) \|) \label{eq:outputBound}
    \end{align}
    for all $\zeta(t) \in \calZ$.
\end{prop}

\begin{proof}
    Let us define $\epsilon(\cdot)$ as in Lemma~\ref{lem:errbound}. 
    Then, it holds that 
    \begin{equation}\label{eq:epsbound}
        \epsilon([y_i^+;\zeta(t)])\le \| [y_i^+;\zeta_i] -[y_i^+;\zeta(t)] \| = \| \zeta_i-\zeta(t) \|.
    \end{equation}
    Therefore,
    \begin{align}\label{eq:gammauProof}
        \left \| u_i - \hat{c}([y_i^+;\zeta(t)]) \right \| &= \left \| c([y_i^+;\zeta_i]) - \hat{c}([y_i^+;\zeta(t)]) \right \| \nonumber\\ 
            &\le \left \| c([y_i^+;\zeta_i]) - c([y_i^+;\zeta(t)]) \right \| + \left \| c([y_i^+;\zeta(t)]) - \hat{c}([y_i^+;\zeta(t)]) \right \| \nonumber\\
            &\le L_c\cdot \left\|\zeta_i-\zeta(t) \right\| + \eta(\|\zeta_i-\zeta(t) \|) \nonumber\\
            &=: \gamma_u(\|\zeta_i-\zeta(t) \|),
    \end{align} 
    where the first inequality follows from the triangle inequality, and the second inequality follows from Assumption~\ref{asm:lipschitz} and Lemma~\ref{lem:errbound} with \eqref{eq:epsbound}.
    Clearly, $\gamma_u(\cdot)$ is a class $\calK_\infty$ function.    
    
    Since $y_i^+=f(\zeta_i,u_i)$, it can be similarly derived that
    \begin{equation*}
        \begin{split}
            \left \| y_i^+ - f(\zeta(t), \hat{c}([y_i^+;\zeta(t)])) \right \| 
            &= \left \| f(\zeta_i,u_i) - f(\zeta(t), \hat{c}([y_i^+;\zeta(t)])) \right \| \\
            &\le \left \| f(\zeta_i,u_i) - f(\zeta(t), u_i) \right \|  + \left \| f(\zeta(t),u_i) - f(\zeta(t), \hat{c}([y_i^+;\zeta(t)])) \right \| \\
            &\le L_f \cdot \left( \left\|\zeta_i-\zeta(t) \right\| + \gamma_u(\left\|\zeta_i-\zeta(t) \right\|) \right) \\
            &=: \gamma_y(\left\|\zeta_i-\zeta(t) \right\|),
        \end{split}
    \end{equation*}
    where $\gamma_y(\cdot)$ is of class $\calK_\infty$ because $\calK_\infty$ functions are closed under addition and positive scalar multiplication.
    This concludes the proof.
\end{proof}

Next, we extend the result of Proposition~\ref{prop:bound} to establish an upper bound on the difference between the corresponding augmented states.
For an augmented state $\zeta(t)$, we define the following vectors that extract all but the oldest input and output:
\begin{align}\label{eq:zetaRe}
    \bfy(t)&:=
    \begin{bmatrix}
        \bfzero_{(n-1)\times 1} &~~ I_{n-1} &~~ \bfzero_{(n-1) \times (n-1)}
    \end{bmatrix}
    \zeta(t)\in\bbR^{n-1},  \nonumber\\
    \bfu(t)&:=
    \begin{bmatrix}
        \bfzero_{(n-2)\times (n+1)} &~~ I_{n-2} 
    \end{bmatrix}
    \zeta(t)\in\bbR^{n-2}.
\end{align}
Then, it follows from \eqref{eq:zeta} that 
\begin{equation}\label{eq:zeta(t+1)}
    \zeta(t+1) =
        \begin{bmatrix}
            \bfy(t)\\ y(t+1) \\ \bfu(t) \\ u(t)
        \end{bmatrix}. 
\end{equation}
Analogously to \eqref{eq:zetaRe}, for each $([y_i^+;\zeta_i],u_i)\in \calD$, we define
\begin{align}\label{eq:zetaiRe}
    \bfy_i&:=
    \begin{bmatrix}
        \bfzero_{(n-1)\times 1} &~~ I_{n-1} &~~ \bfzero_{(n-1) \times (n-1)}
    \end{bmatrix}
    \zeta_i\in\bbR^{n-1}, \nonumber \\
    \bfu_i&:=
    \begin{bmatrix}
        \bfzero_{(n-2)\times (n+1)} &~~ I_{n-2} 
    \end{bmatrix}
    \zeta_i\in\bbR^{n-2},
\end{align}
and $\zeta_i^+\in\bbR^{2n-1}$ that corresponds to $\zeta(t+1)$ in \eqref{eq:zeta(t+1)}:
\begin{equation}\label{eq:zeta+}
    \zeta_i^+ := \begin{bmatrix}
         \bfy_i\\ y_i^+ \\ \bfu_i \\ u_i
    \end{bmatrix}.
\end{equation}

The following proposition states that $\|\zeta_i^+ - \zeta(t+1) \|$ is also bounded above by a class $\calK_\infty$ function of $\|\zeta_i-\zeta(t) \|$.

\begin{prop}\label{prop:zetabound}\upshape
    For the class $\calK_\infty$ function $\gamma:\bbR_{\ge 0}\ra\bbR_{\ge 0}$, defined by
    \begin{equation}\label{eq:gammaDef}
        \gamma(\epsilon) := \gamma_u(\epsilon) + \gamma_y(\epsilon) + \epsilon, 
    \end{equation}
     controller \eqref{eq:igp} with the reference point $y_\rmr(t+1)=y_i^+$ guarantees
    \begin{align}\label{eq:zetaboundToShow}
        \| \zeta_i^+ - \zeta(t+1) \| \le \gamma(\|\zeta_i - \zeta(t) \|)
    \end{align}
    for any $([y_i^+;\zeta_i],u_i)\in \calD$ and $\zeta(t)\in\calZ$.
\end{prop}
\begin{proof}
    Applying the triangle inequality, we have 
        \begin{equation*}
            \| \zeta_i^+ - \zeta(t+1) \| \le \| y_i^+ - y(t+1) \| + \| u_i - u(t)\| + \left\| \begin{bmatrix}
                \bfy(t) - \bfy_i \\
                \bfu(t) - \bfu_i
            \end{bmatrix} \right\|. 
        \end{equation*}
        Then, it is straightforward from Proposition~\ref{prop:bound} that \eqref{eq:zetaboundToShow} holds, which concludes the proof.
\end{proof}

For a compact set $\calA\subset\bbR^{2n-1}$, let us define
\begin{align*}\label{eq:k=1cond}
    \calI(\calA):= \left\{ (i, r_i) \in \{1, \dots, N\} \times \bbR_{\ge 0} \mid 
    ([y_i^+; \zeta_i],\, u_i) \in \calD,~\zeta_i^+ \in \calA, ~
    r_i := \max\{ r \ge 0 \mid \calB(\zeta_i^+, r) \subset \calA \} \right\},
\end{align*}
where $r_i$ is well-defined since $\calA$ is compact.
Then, for each $(i,r_i)\in\calI(\calA)$, it follows from Proposition~\ref{prop:zetabound} that if
\begin{equation}\label{eq:zetaInBall}
    \zeta(t) \in \calB(\zeta_i,\gamma^{-1}(r_i))
\end{equation}
then controller \eqref{eq:igp} with the reference point $y_\rmr(t+1) = y_i^+$ guarantees $\|\zeta_i^+ - \zeta(t+1)\|\le r_i$, which implies
\begin{align}\label{eq:calIProp}
    \zeta(t+1) \in \calB(\zeta_i^+, r_i)\subset \calA.
\end{align}
This observation leads to the following lemma, which provides a method to obtain an underestimate of $\calR^{-1}(\calA)$ defined in \eqref{eq:backReach} from $\calD$.

\begin{lem}\label{lem:backward}\upshape
    Given the dataset $\calD$ and a compact set $\calA$,
    \begin{equation}\label{eq:rHat}
        \hat{\calR}^-(\calA) := \bigcup_{(i,r_i) \in \calI(\calA)} \calB(\zeta_i, \gamma^{-1}(r_i))
    \end{equation}
    is an underestimate of $\calR^-(\calA)$, i.e., $\hat{\calR}^-(\calA)\subset \calR^-(\calA)$.
\end{lem}


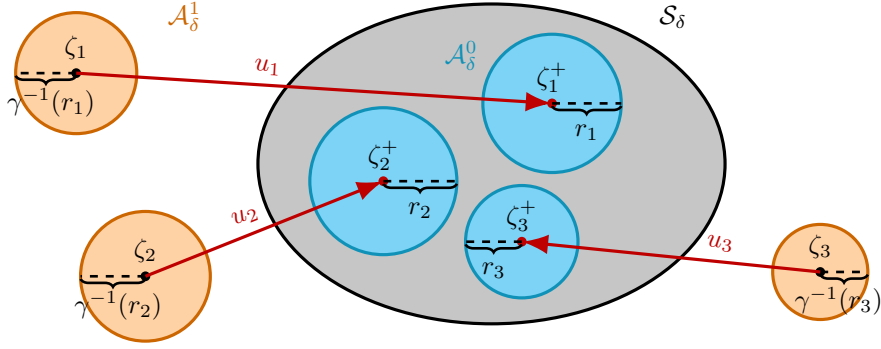
\begin{figure}[t]
\centering
\begin{tikzpicture}[scale=1.15]
\definecolor{myred}{rgb}{0.75,0.0,0.0}

\draw[black, very thick, fill=gray!60, fill opacity=0.75]
  (0,0) ellipse (2.70 and 1.85);
\node[black] at (2.10,1.70) {$\calS_\delta$};

\coordinate (zp1) at (0.70,0.7);
\coordinate (zp2) at (-1.25,-0.20);
\coordinate (zp3) at (0.35,-0.90);

\def\rpOne{0.80}
\def\rpTwo{0.85}
\def\rpThree{0.65}

\filldraw[fill=cyan!45, draw=cyan!70!black, very thick] (zp1) circle (\rpOne);
\filldraw[fill=cyan!45, draw=cyan!70!black, very thick] (zp2) circle (\rpTwo);
\filldraw[fill=cyan!45, draw=cyan!70!black, very thick] (zp3) circle (\rpThree);

\node[cyan!70!black] at (-0.35,1.25) {$\calA_\delta^{0}$};

\fill[myred] (zp1) circle (1.6pt);
\fill[myred] (zp2) circle (1.6pt);
\fill[myred] (zp3) circle (1.6pt);

\node[black] at ($(zp1)+(0,0.30)$) {$\zeta_{1}^{+}$};
\node[black] at ($(zp2)+(0,0.3)$) {$\zeta_{2}^{+}$};
\node[black] at ($(zp3)+(0,0.3)$) {$\zeta_{3}^{+}$};

\draw[black, dashed, line width=1pt] (zp1) -- ($(zp1)+(\rpOne,0)$);
\draw[decorate, decoration={brace, amplitude=3pt, raise=1pt, mirror}, line width=1pt, black]
  (zp1) -- ($(zp1)+(\rpOne,0)$)
  node[midway, yshift=-11pt] {$r_1$};

\draw[black, dashed, line width=1pt] (zp2) -- ($(zp2)+(\rpTwo,0)$);
\draw[decorate, decoration={brace, amplitude=3pt, raise=1pt, mirror}, line width=1pt, black]
  (zp2) -- ($(zp2)+(\rpTwo,0)$)
  node[midway, yshift=-11pt] {$r_2$};
  
\draw[black, dashed, line width=1pt] (zp3) -- ($(zp3)+(-\rpThree,0)$);
\draw[decorate, decoration={brace, amplitude=3pt, raise=1pt}, line width=1pt, black]
  (zp3) -- ($(zp3)+(-\rpThree,0)$)
  node[midway, yshift=-11pt] {$r_3$};
  
\coordinate (z1) at (-4.80,1.05);
\coordinate (z2) at (-4.00,-1.30);
\coordinate (z3) at (3.80,-1.25);

\def\rOne{0.70}
\def\rTwo{0.75}
\def\rThree{0.55}

\filldraw[fill=orange!35, draw=orange!80!black, very thick] (z1) circle (\rOne);
\filldraw[fill=orange!35, draw=orange!80!black, very thick] (z2) circle (\rTwo);
\filldraw[fill=orange!35, draw=orange!80!black, very thick] (z3) circle (\rThree);

\node[orange!80!black] at (-3.55,1.7) {$\calA_\delta^{1}$};

\fill[black] (z1) circle (1.6pt);
\fill[black] (z2) circle (1.6pt);
\fill[black] (z3) circle (1.6pt);

\node[black] at ($(z1)+(0,0.3)$) {$\zeta_{1}$};
\node[black] at ($(z2)+(0,0.3)$) {$\zeta_{2}$};
\node[black] at ($(z3)+(0,0.3)$) {$\zeta_{3}$};

\draw[black, dashed, line width=1pt] (z1) -- ($(z1)+(-\rOne,0)$);
\draw[decorate, decoration={brace, amplitude=3pt, raise=1pt}, line width=1pt, black]
  (z1) -- ($(z1)+(-\rOne,0)$)
  node[midway, xshift=2pt, yshift=-11pt] {$\gamma^{-1}(r_1)$};

\draw[black, dashed, line width=1pt] (z2) -- ($(z2)+(-\rTwo,0)$);
\draw[decorate, decoration={brace, amplitude=3pt, raise=1pt}, line width=1pt, black]
  (z2) -- ($(z2)+(-\rTwo,0)$)
  node[midway, xshift=2pt, yshift=-11pt] {$\gamma^{-1}(r_2)$};

\draw[black, dashed, line width=1pt] (z3) -- ($(z3)+(\rThree,0)$);
\draw[decorate, decoration={brace, amplitude=3pt, raise=1pt, mirror}, line width=1pt, black]
  (z3) -- ($(z3)+(\rThree,0)$)
  node[midway, xshift=-2pt, yshift=-11pt] {$\gamma^{-1}(r_3)$};

\draw[-{Latex[length=3.5mm]}, line width=1.2pt, myred]
  (z1) -- (zp1) node[midway, above, sloped, xshift=-18pt] {$u_1$};

\draw[-{Latex[length=3.5mm]}, line width=1.2pt, myred]
  (z2) -- (zp2) node[midway, above, sloped,xshift=-5pt] {$u_2$};

\draw[-{Latex[length=3.5mm]}, line width=1.2pt, myred]
  (z3) -- (zp3) node[midway, above, sloped, xshift=18pt] {$u_3$};

\end{tikzpicture}
\caption{Abstract illustration of $\calS_\delta$ (gray area), $\calA_\delta^{0}$ (cyan area), and $\calA_\delta^{1}$ (orange area). 
For each $i\in\{1,2,3\}$, $\zeta_i^+$ is given by \eqref{eq:zeta+} for $([y_i^+;\zeta_i],u_i)\in\calD$ and $(i,r_i)\in\calI(\calS_\delta)$.}
\label{fig:diagram}
\end{figure}

Using \eqref{eq:calIProp} and the result of Lemma~\ref{lem:backward}, the following proposition explicitly constructs the sequence of sets $(\calA_\delta^j)_{j=0}^\infty$ satisfying the defining property given in \eqref{eq:defProp}.

\begin{prop}\label{prop:calA}\upshape
    Given the dataset $\calD$ and $\delta>0$, let the sequence $(\calA_\delta^j)_{j=0}^\infty$ be recursively constructed as
    \begin{equation}\label{eq:expAdef}
        \begin{split}
            \calA_\delta^0 &= \bigcup_{(i,r_i) \in  \calI(\calS_\delta)}\calB(\zeta_i^+, r_i),  ~~~~~
            \calA_\delta^{j+1}=\hat{\calR}^-(\calA_\delta^j)~~~\forall j\in\bbZ_{\ge 0}.
        \end{split}
    \end{equation}
    Then, \eqref{eq:expAdef} is consistent with the defining property \eqref{eq:defProp}.
\end{prop}
\begin{proof}
    It follows from the definition of $r_i$ that $\calB(\zeta_i^+,r_i)\subset \calS_\delta$ for each $(i,r_i)\in\calI(\calS_\delta)$, which implies $\calA_\delta^0 \subset \calS_\delta$.
    Also, it directly follows from Lemma~\ref{lem:backward} that $\calA_\delta^{j+1}=\hat{\calR}^-(\calA_\delta^j)\subset\calR^{-}(\calA_\delta^j)$ for all $j\in\bbZ_{\ge 0}$, and this concludes the proof.
\end{proof}

Figure~\ref{fig:diagram} provides an abstract illustration of $\calS_\delta$, and $\calA_\delta^{0}$ and $\calA_\delta^{1}$ as defined in \eqref{eq:expAdef}.
Note that \eqref{eq:expAdef} can be entirely constructed from the dataset $\calD$, since the function $\gamma$ in Proposition~\ref{prop:zetabound} is known under Assumption~\ref{asm:RKHSnorm}.
This allows one to collect $\calD$ and pre-compute \eqref{eq:expAdef} during the offline procedure.
The subsequent theorem provides a method to choose an appropriate sequence of reference points $y_{i_1}^+, y_{i_2}^+, \ldots, y_{i_\kappa}^+ \in\calD_\sfto$ described in Lemma~\ref{lem:abskstep}, given \eqref{eq:expAdef}.
\begin{thm}\label{thm:refk>0}\upshape
    Given the dataset $\calD$ and $\delta>0$, consider the sequence of sets $(\calA_\delta^j)_{j=0}^\infty$ given by \eqref{eq:expAdef}.
    If $\zeta(t)\in\calA_\delta^\kappa$ for some $\kappa\in\bbN$, then there exist 
    $(i_l,r_{i_l})\in \calI(\calA_\delta^{\kappa-l})$ for $l=1,\ldots, \kappa$
    such that 
    \begin{align}
        \zeta(t) &\in \calB(\zeta_{i_1},\gamma^{-1}(r_{i_1})),\label{eq:cor2ToShow1} \\
        \calB(\zeta_{i_l}^+, r_{i_l})  &\subset \calB(\zeta_{i_{l+1}},\gamma^{-1}(r_{i_{l+1}})) ~~~ \forall l=1,\ldots, \kappa-1. \label{eq:cor2ToShow2}
    \end{align}
    Then, controller $\eqref{eq:igp}$ with the reference points $y_\rmr(t+l) = y_{i_l}^+$ for $l=1,\ldots,\kappa$ guarantees 
    \begin{equation}\label{eq:expkstep}
        \zeta(t+l)\in\calA_\delta^{\kappa-l} \quad \forall l=1,\ldots,\kappa,
    \end{equation}
    which implies $\|y(t+\kappa)\|\le \delta$.
\end{thm}

\begin{proof}
    Since $\zeta(t)\in\calA_\delta^\kappa = \hat{\calR}^-(\calA_\delta^{\kappa-1})$, it follows from \eqref{eq:rHat} that there exists $(i_1,r_{i_1})\in\calI(\calA_\delta^{\kappa-1})$ such that \eqref{eq:cor2ToShow1} holds.
    Then, similar to \eqref{eq:calIProp}, it is guaranteed that 
    \begin{equation*}
        \zeta(t+1) \in \calB(\zeta_{i_1}^+,r_{i_1})\subset \calA_\delta^{\kappa-1} = \hat{\calR}^-(\calA_\delta^{\kappa-2}),
    \end{equation*}
    where the last equality implies the existence of $(i_2,r_{i_2})\in\calI(\calA_\delta^{\kappa-2})$ such that \eqref{eq:cor2ToShow2} holds for the case $l=1$.
    By repeating this reasoning, we can recursively select $(i_l,r_{i_l})$ for $l=3,\ldots, \kappa$ such that \eqref{eq:cor2ToShow2} holds, and this concludes the proof.
\end{proof}


\begin{algorithm}[t]
    \caption{Inverse learning-based control}\label{alg:pgigp}
    \begin{algorithmic}[1]
    \Require $\calD$, $\bar{\kappa}\in\bbN$, and $\delta>0$
    \renewcommand{\algorithmicrequire}{\textbf{\#Offline}}
    \Require
    \State Compute $(\calA_\delta^j)_{j=0}^{\bar{\kappa}}$ in \eqref{eq:expAdef}\label{Line:while}
    \renewcommand{\algorithmicrequire}{\textbf{\#Online}}
    \Require 
    \State $t\gets 0$ 
    \State $\kappa=\min\{0\le j\le\bar{\kappa}\mid \zeta(t)\in\calA_\delta^{j}\}$ \label{line:return}
    \State Find $(i_1,r_{i_1})$ according to Theorem~\ref{thm:refk>0}
    \State Set $y_\rmr(t+1) \leftarrow y_{i_1}^+$
    \State Apply $u(t) \gets \hat{c} ( [y_\rmr(t+1);\zeta(t)])$
    \State $t\gets t+1$ and go to line~\ref{line:return}
    \end{algorithmic}
\end{algorithm}


Finally, the following theorem presents a (data-dependent) sufficient condition under which the inverse learning-based controller \eqref{eq:igp} with appropriate reference points from $\calD_\sfto$ guarantees \eqref{eq:goal}.
We omit the proof since it directly follows from Lemma~\ref{lem:absmain} and Theorem~\ref{thm:refk>0}.

\begin{thm}\label{thm:main}\upshape
    Given the dataset $\calD$ and $\delta>0$, consider the sequence of sets $(\calA_\delta^j)_{j=0}^\infty$ given by \eqref{eq:expAdef}.
    If
    \begin{equation}\label{eq:suff}
        \zeta(0) \in \calA_\delta^\kappa \ \mbox{and} \ \calA_\delta^0 \subset \calA_\delta^1
    \end{equation}
    for some $\kappa\in\bbN$ then controller \eqref{eq:igp}, with the reference points selected from $\calD_\sfto$ according to Theorem~\ref{thm:refk>0}, guarantees \eqref{eq:goal}.
\end{thm}


Algorithm~\ref{alg:pgigp} presents a practical guideline for implementing the proposed controller.
In practice, computing the infinitely many sets in $(\calA_\delta^j)_{j=0}^\infty$ is computationally infeasible.
To address this issue, we introduce an empirical upper bound $\bar{\kappa}\in\bbN$ and compute only $(\calA_\delta^j)_{j=0}^{\bar{\kappa}}$.
If the second condition $\calA_\delta^0 \subset \calA_\delta^1$ in \eqref{eq:suff} is not satisfied, we suggest augmenting $\calD$ with additional data until the condition holds, before executing Line~\ref{Line:while}.
At each time step $t\in\bbZ_{\ge 0}$, we determine the smallest $j$ with $0\le j \le \bar{\kappa}$ such that $\zeta(t)\in\calA_\delta^j$ and use it to select the reference point, thereby ensuring practical output regulation within fewer steps.


\begin{rem}\upshape
    We comment on the memory and computational efficiency of Algorithm~\ref{alg:pgigp}.
    Implementing the algorithm is not memory intensive since the sets $(\calA_\delta^j)_{j=0}^{\bar{\kappa}}$ are fully characterized by the finite tuples $\calI(\calS_\delta)$ and $(\calI(\calA_\delta^j))_{j=0}^{\bar{\kappa}-1}$, which can be pre-computed during the offline procedure. 
    This significantly reduces the online computation burden, as one simply has to check set inclusions, as seen in Theorem~\ref{thm:refk>0} and Line~\ref{line:return} of Algorithm~\ref{alg:pgigp}.
\end{rem}

\begin{rem}\upshape\label{rem:mimo}
    The proposed framework can be readily extended to multi-input multi-output (MIMO) case, where $m\ge 1$ and $p\ge1$.
    For notational convenience, we define $\calC := \bbR^{(n+1)p + (n-1)m}$. 
    In the MIMO case, the inverse model of the system \eqref{eq:sys} can be defined analogously to Definition~\ref{def:inverse}, where the inverse model $c:\calC \ra \bbR^m$ is now a multi-output function that satisfies \eqref{eq:inv} for all $\zeta \in \calZ\subset \bbR^{np + (n-1)m}$ and $y^+\in\calR(\zeta)\subset \bbR^p$.
    Given the dataset $\calD$ in \eqref{eq:dataset}, where $y_i^+\in \bbR^{p}$, $\zeta_i\in \bbR^{np + (n-1)m}$ and $u_i\in\bbR^m$, let the inverse model $c$ and the training output $u_i$ for $i=1,\ldots,N$ be decomposed as follows:
    \begin{equation*}
        c(\cdot) =: 
        \begin{bmatrix}
            c^1(\cdot) \\
            \vdots \\
            c^m(\cdot)
        \end{bmatrix} , \quad 
        u_i=: 
        \begin{bmatrix}
            u_i^1 \\
            \vdots \\
            u_i^m
        \end{bmatrix},
    \end{equation*}
    where $c^j:\calC \ra \bbR$ and $u_i^j \in \bbR$ for $j=1,\ldots,m$.
    Since 
    \begin{equation*}
        u_i^j = c^j([y_i^+;\zeta_i]) ,
    \end{equation*}
    we can obtain a KI estimate $\hat{c}^j:\calC \ra \bbR$ for each single-output function $c^j$, following the steps described in Section~\ref{subsec:inv}.
    Here, different kernel functions may be used for obtaining each estimate $\hat{c}^j$, which offers flexibility to capture heterogeneous characteristics of the output channels.
    By defining $\hat{c}:=[\hat{c}^1; \cdots ;\hat{c}^m]$, the analyses in Sections~\ref{subsec:err} and~\ref{sec:control} can be repeated.
    Alternatively, one may directly obtain an estimate of $c$ using vector-valued kernel methods, which are capable of capturing correlations among the single-output functions $c^j$; see, e.g., \cite{MiccPont05} and \cite{AlvaRosa12}.
\end{rem}

\subsection{NARX models with input delays}\label{subsec:extend}

We illustrate how our proposed framework can be extended to NARX models with input delays.
Suppose that system \eqref{eq:sysNARX} has an input delay $\nu\in\bbN$ such that $2\le \nu\le n$. 
This means that the system has a relative degree of at least $\nu$, i.e., it takes at least $\nu$ steps for the input to explicitly affect the output.
For notational simplicity, we focus on the case $\nu=2$ and rewrite system \eqref{eq:sysNARX} as
\begin{align}\label{eq:sysDelay}
    y(t+2) = f(y_{[t-n+2,t+1]}, u_{[t-n+2,t]}),
\end{align}
noting that the following derivations can be extended to arbitrary $\nu$ in a straightforward manner. 
We often write \eqref{eq:sysDelay} compactly as $y(t+2) = f(\zeta(t+1))$.

Since
\begin{equation}\label{eq:y+zeta}
    y(t+1) = f(\zeta(t)),
\end{equation}
substituting this into \eqref{eq:sysDelay} allows us to express $y(t+2)$ as
\begin{align}\label{eq:deffbar}
    y(t+2) &= f([y_{[t-n+2,t]};y(t+1)],u_{[t-n+2,t]}) \nonumber \\
    &=: \bar{f}(\zeta(t),u(t))
\end{align}
for some function $\bar{f}:\bbR^{2n-1}\times \bbR \to \bbR$. 
This form is identical to \eqref{eq:sys}, except that the input $u(t)$ now does not affect $y(t+1)$ but affects $y(t+2)$, which is two-steps ahead.

In what follows, we redefine the notations from Sections~\ref{sec:prob}-\ref{sec:ilc} with respect to system \eqref{eq:deffbar}.
First, we adapt \eqref{eq:reachable} and define the \textit{two-step reachable set of outputs} for any $\zeta\in\calZ$ as 
\begin{equation*}
    \calR(\zeta):= \left\{y^{++} \in \bbR \mid \exists u\in\bbR \ \mbox{such that} \ y^{++}=\bar{f}(\zeta,u)  \right\},
\end{equation*}
and reformulate Assumption~\ref{asm:inverse} as follows.

\begin{asm}\upshape\label{asm:delayInjective}
    For any $\zeta\in\calZ$ and $y^{++}\in\calR(\zeta)$, there exists a unique $u\in\bbR$ such that $y^{++}=\bar{f}(\zeta,u)$.
\end{asm}

Assumption~\ref{asm:delayInjective} ensures that system \eqref{eq:deffbar} has a global relative degree two.
Under this assumption, we define the inverse model of \eqref{eq:deffbar} as a function $c:\bbR^{2n}\to \bbR$ that satisfies
\begin{equation}\label{eq:inverseDelay}
    y^{++} = \bar{f}(\zeta,c([y^{++};\zeta]))
\end{equation}
for all $\zeta\in\calZ$ and $y^{++}\in\calR(\zeta)$.
Accordingly, we propose an inverse learning-based controller of the form 
\begin{equation}\label{eq:igp2}
    u(t) = \hat{c}([y_\rmr(t+2);\zeta(t)]),
\end{equation}
whose objective is to make the two-steps ahead output $y(t+2)$ track a desired reference $y_\rmr(t+2)$.

In order to obtain an estimate $\hat{c}$ of the inverse model $c$ in this case, we slightly modify the dataset $\calD$ given in \eqref{eq:dataset}, as
\begin{equation}\label{eq:dataset2}
    \calD:=\{([y_i^{++};\zeta_i],u_i) \}_{i=1}^N,
\end{equation}
where the number of training data is now $N=T-n$ and $y_i^{++}:=y^\rmd(i+n)$, so that $y_i^{++}=\bar{f}(\zeta_i,u_i)$ for all $i=1,\ldots,N$. 
We also redefine $\xi_i:=[y_i^{++};\zeta_i]$ and $\calD_\sfto:=\{y_i^{++}\}_{i=1}^N$.
The estimate $\hat{c}$ is then obtained by solving \eqref{eq:krr} with respect to the dataset \eqref{eq:dataset2}.

Given $\hat{c}$, the functions $\gamma_u$ in Proposition~\ref{prop:bound} and $\gamma$ in Proposition~\ref{prop:zetabound} can be computed as follows.

\begin{prop}\label{prop:bound2}\upshape
    Consider system \eqref{eq:sysDelay} and the dataset $\calD$ in \eqref{eq:dataset2}.
    There exists a known class $\mathcal{K}_\infty$ function $\gamma_u: \bbR_{\ge 0} \ra \bbR_{\ge 0}$ such that for any $ ([y_i^{++};\zeta_i],u_i)\in \calD$, 
    \begin{align}\label{eq:bound2}
        \left \| u_i - \hat{c}([y_i^{++};\zeta(t)]) \right \| &\le \gamma_u(\|\zeta_i-\zeta(t) \|),
    \end{align}
    for all $\zeta(t) \in \calZ$.
\end{prop}
\begin{proof}
    Adapting the definition of $\epsilon(\cdot)$ given in Lemma~\ref{lem:errbound} with respect to the redefined $\xi_i$, we obtain $\epsilon([y_i^{++};\zeta(t)]) \le \|\zeta_i-\zeta(t) \|$.
    The proof is then completed by following the same steps as in \eqref{eq:gammauProof}.
\end{proof}

\begin{prop}\upshape\label{prop:gamma2}
    Consider system \eqref{eq:sysDelay} and the dataset $\calD$ in \eqref{eq:dataset2}.
    For the class $\calK_\infty$ function $\gamma:\bbR_{\ge 0}\ra\bbR_{\ge 0}$ defined by
    \begin{equation}\label{eq:gammaDef2}
        \gamma(\epsilon) := \gamma_u(\epsilon) + (1+L_f) \epsilon, 
    \end{equation}
    controller \eqref{eq:igp2} with the reference point $y_\rmr(t+2)=y_i^{++}$ guarantees
    \begin{align}\label{eq:zetaboundToShow2}
        \| \zeta_i^+ - \zeta(t+1) \| \le \gamma(\|\zeta_i - \zeta(t) \|)
    \end{align}
    for any $([y_i^{++};\zeta_i],u_i)\in \calD$ and $\zeta(t)\in\calZ$.
\end{prop}


\begin{proof}
    Having \eqref{eq:zeta(t+1)}, \eqref{eq:zeta+} and \eqref{eq:y+zeta} in mind, observe that 
    \begin{align*}
        \zeta_i^+ - \zeta(t+1)  
        &= \begin{bmatrix}
            \bfy_i \\
            y_i^+ \\
            \bfu_i \\
            u_i
        \end{bmatrix} - 
        \begin{bmatrix}
            \bfy(t) \\
            y(t+1) \\ 
            \bfu(t) \\
            u(t)
        \end{bmatrix} 
        = 
        \begin{bmatrix}
            \bfy_i \\
            f(\zeta_i) \\
            \bfu_i \\
            u_i
        \end{bmatrix} - 
        \begin{bmatrix}
            \bfy(t) \\
            f(\zeta(t)) \\ 
            \bfu(t) \\
            u(t)
        \end{bmatrix}
    \end{align*}
    By the triangle inequality, Lipschitz continuity of $f$, and Proposition~\ref{prop:bound2}, we have 
    \begin{align*}
         \left\|\zeta_i^+-\zeta(t+1) \right\|  
        &\le \left\|\begin{bmatrix}
            \bfy_i - \bfy(t) \\
            \bfu_i - \bfu(t)
        \end{bmatrix}  \right\| + \left\| f(\zeta_i) - f(\zeta(t)) \right\| + \left\| u_i - u(t) \right\| \\ 
        &\le (1+L_f)\left\|\zeta_i-\zeta(t) \right\| + \gamma_u(\|\zeta_i-\zeta(t)\|) \\
        &=: \gamma(\|\zeta_i-\zeta(t)\|).
    \end{align*}
    and this concludes the proof.
\end{proof}

Proposition~\ref{prop:gamma2} explicitly characterizes the upper bound function $\gamma$ with respect to the input delayed system \eqref{eq:sysDelay}.
Consequently, all subsequent analyses and results established in Section~\ref{sec:ilc} can be repeated by employing the redefined $\gamma$ in \eqref{eq:gammaDef2}.
 
\section{Simulation Results}\label{sec:sim}

This section provides simulation results\footnote{The code is fully available at \url{https://github.com/yj-jang-98/Inv_Learn_Ctrl}.} to demonstrate the effectiveness of the proposed inverse learning-based controller \eqref{eq:igp}.
The established theoretical guarantees are first validated through a numerical example. 
We then apply the proposed controller to the stabilization problem of an inverted pendulum to demonstrate its practical applicability under realistic conditions. 
In particular, we evaluate the performance of the proposed controller under additive measurement noise and compare the resulting closed-loop behavior against the noise-free case.

\subsection{Numerical example}\label{subsec:numerical}

Let system \eqref{eq:sys} be given as 
\begin{align} \label{eq:sysNumerical}
    y(t+1)
    &= -3 + \sqrt{- \|\zeta(t)\|_2^2 - 16 \ln(u(t))} 
\end{align}
with $n=2$ and $m=p=1$,
where the input $u(t)\in\bbR$ and output $y(t)\in\bbR$ are subject to the constraints 
\begin{align}\label{eq:inputConst}
    u(t) & \in \calU(\zeta(t)) := \left\{ u\in \bbR \mid 4 \le - \left\|\zeta(t) \right\|_2^2 - 16\ln(u) \le 16 \right\},\\
    y(t) & \in \calY := [-1,1]. \nonumber
\end{align}
We employed the widely used squared-exponential kernel \cite{WillRasm06} defined by 
\begin{align*}
    k(\xi,\xi') = \exp \left( -\frac{\|\xi-\xi'\|_2^2}{2\sigma_l^2} \right),
\end{align*}
which is strictly positive definite, isotropic, and decreasing. 
We set the length-scale hyperparameter to $\sigma_l = 2\sqrt{2}$.

First, we verify that Assumptions~\ref{asm:inverse}--\ref{asm:RKHSnorm} hold for system \eqref{eq:sysNumerical}.
The input/output constraints \eqref{eq:inputConst} imply that the feasible set of augmented states is $\calZ = \calY \times \calY \times \calU$ .
Moreover, the input constraint ensures that if $\zeta(t)\in\calZ$ and $u(t)\in\calU(\zeta(t))$, then $y(t+1)$ is well-defined and satisfies $y(t+1) \in \calY$.
Consequently, the one-step reachable set of outputs is given by $\calR(\zeta)=\calY$ for all $\zeta\in\calZ$.
Since $f(\zeta,u)$ is strictly decreasing in $u$ for any fixed $\zeta \in \calZ$, Assumption~\ref{asm:inverse} holds, and the corresponding inverse model is obtained as 
\begin{align*}
    c([y^+;\zeta])
    = \exp\left(-\frac{(y^+ + 3)^2 + \|\zeta\|_2^2}{16}\right)
    = k\left([y^+;\zeta],[-3;0;0;0]\right).
\end{align*}
Therefore, $c\in\mathcal{H}_k$ with $\|c\|_{\mathcal{H}_k}=1$; hence, we set $\Gamma=1$ in Assumption~\ref{asm:RKHSnorm}.
While Assumption~\ref{asm:lipschitz} is stated globally, it suffices to compute the Lipschitz constants $L_f$ and $L_c$ locally over the compact domain induced by the constraints in \eqref{eq:inputConst}.
Direct calculation yields $L_f = 6.5$ and $L_c = 0.22$, and as a result, the functions $\eta$ and $\gamma$ in Lemma~\ref{lem:errbound} and Proposition~\ref{prop:zetabound}, respectively, can be computed as 
\begin{align*}
    \eta(\epsilon) &= \sqrt{1- \exp\left(-\frac{\epsilon^2}{16}\right)}, \\
    \gamma(\epsilon) &= \left(L_c + L_f + L_fL_c +1  \right)\epsilon + \left( 1+L_f\right) \eta(\epsilon). 
\end{align*}

The training dataset $\calD$ was generated by performing multiple one-step experiments on a uniform grid of initial conditions.
Specifically, we constructed initial augmented states of the form $\zeta^\rmd(0)= [y^\rmd(-1);y^\rmd(0);u^\rmd(-1)]$
by selecting $y^\rmd(-1)$ and $y^\rmd(0)$ from a uniform grid over $\calY$ consisting of seven points, and selecting $u^\rmd(-1)$ from a uniform grid over $[0,1]$ with four points.
For each such initial condition, we then applied ten random inputs $u^\rmd(0)\in\calU(\zeta^\rmd(0))$ to obtain $\zeta^\rmd(1)$, yielding a total of $N=280$ training data.

For the implementation of Algorithm~\ref{alg:pgigp}, some practical issues arise. 
First, the sets $\calA_\delta^j$ cannot be computed for infinitely many $j\in\bbZ_{\ge 0}$ in practice. 
Second, for a given $\delta$, there may not exist a $\kappa\in\bbZ_{\ge 0}$ for which the sufficient condition \eqref{eq:suff} holds.
To address these issues, we set $\bar{\kappa}=20$ and consider the set of candidate accuracy levels $\Delta=\{0.1,0.2,0.3,0.4,0.5,1,1.5,2,3 \}$.
Then, for each $\delta\in\Delta$, we precomputed $(\calA_\delta^j)_{j=0}^{\bar{\kappa}}$.
When executing Line~\ref{Line:while} of Algorithm~\ref{alg:pgigp}, a pair of $0\le \kappa \le \bar{\kappa}$ and $\delta\in\Delta$ that satisfies \eqref{eq:suff} is selected, preferring smaller $\delta$ and then smaller $\kappa$.

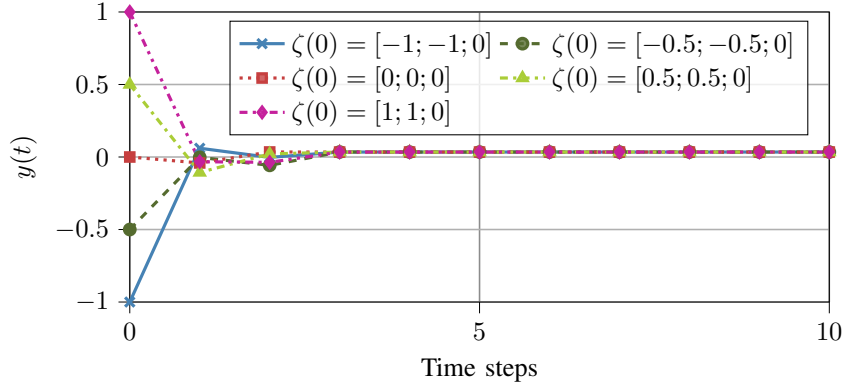
\begin{figure}[t]
    \centering
\begin{tikzpicture}

\definecolor{darkgray176}{RGB}{176,176,176}
\definecolor{darkolivegreen}{RGB}{85,107,47}
\definecolor{greenyellow19622665}{RGB}{170,205,55}
\definecolor{indianred1986363}{RGB}{198,63,63}
\definecolor{mediumvioletred19733156}{RGB}{197,33,156}
\definecolor{steelblue}{RGB}{70,130,180}

\begin{axis}[
width=0.6\textwidth,
height=0.3\textwidth,
tick align=outside,
tick pos=left,
xlabel={Time steps},
xmin=0, xmax=10,
xtick style={color=black},
xtick={0,5,10},
y grid style={darkgray176},
ylabel={$y(t)$},
ymin=-1, ymax=1,
ytick style={color=black},
xmajorgrids,
ymajorgrids,
grid style={line width=0.6pt, draw=black!45},
legend pos= north east,
legend cell align=left,
legend columns=2,
legend style={draw=black, fill=white, fill opacity=0.8, text opacity=1},
]
\addplot[
  line width=1.2pt, steelblue,
  mark=x, mark options={solid},
  mark size=2.5pt
]
table {%
0 -1
1 0.0604664450568415
2 -0.00251995029169416
3 0.0340113876004113
4 0.0340996526333384
5 0.0341810109649932
6 0.0341800151467648
7 0.0341813958501871
8 0.0341814198847175
9 0.0341809566669329
10 0.0341809566669329
};
\addlegendentry{$\zeta(0)=[-1;-1;0]$}

\addplot [line width=1.2pt,
  mark=*, mark options={solid},
  mark size=2pt, darkolivegreen, dashed]
table {%
0 -0.5
1 -0.00247637944630119
2 -0.057506052233097
3 0.0345413311574192
4 0.034092698772846
5 0.0341829596072012
6 0.0341816055441346
7 0.0341795754239111
8 0.0341816701509483
9 0.0341824485084432
10 0.0341824485084432
};
\addlegendentry{$\zeta(0)=[-0.5;-0.5;0]$}

\addplot [line width=1.2pt,
  mark=square*, mark options={solid},
  mark size=1.6pt, indianred1986363, dotted]
table {%
0 0
1 -0.0410664812414052
2 0.0347266107472888
3 0.0340591719538779
4 0.0341836600398131
5 0.0341808838218882
6 0.034180143596044
7 0.0341811875338793
8 0.0341812839714768
9 0.0341807650029735
10 0.0341807650029735
};
\addlegendentry{$\zeta(0)=[0;0;0]$}

\addplot [line width=1.2pt,
  mark=triangle*, mark options={solid},
  mark size=2pt, greenyellow19622665, dashdotted]
table {%
0 0.5
1 -0.106530242175488
2 0.0212823851621247
3 0.0341614742438514
4 0.0341371176567784
5 0.034181220603553
6 0.0341798928364203
7 0.0341813997052056
8 0.0341820568958129
9 0.0341810372091089
10 0.0341810372091089
};
\addlegendentry{$\zeta(0)=[0.5;0.5;0]$}

\addplot [line width=1.2pt,
  mark=diamond*, mark options={solid},
  mark size=2pt, mediumvioletred19733156, dashdotdotted]
table {%
0 1
1 -0.0325871093534196
2 -0.0366388124527819
3 0.033978815131547
4 0.0340843118389325
5 0.034180208078888
6 0.034180087254402
7 0.0341814113408261
8 0.0341813553302472
9 0.0341809548002363
10 0.0341809548002363
};
\addlegendentry{$\zeta(0)=[1;1;0]$}

\end{axis}

\end{tikzpicture}
    \caption{Output trajectories generated by the proposed controller with initial conditions $[-1;-1;0]$ (blue solid line), $[-0.5;-0.5;0]$ (green dashed line), $[0;0;0]$ (red dotted line), $[0.5;0.5;0]$ (yellow dash-dotted line), and $[1;1;0]$ (violet dash-dot-dotted line). }
    \label{fig:traj1}
\end{figure}

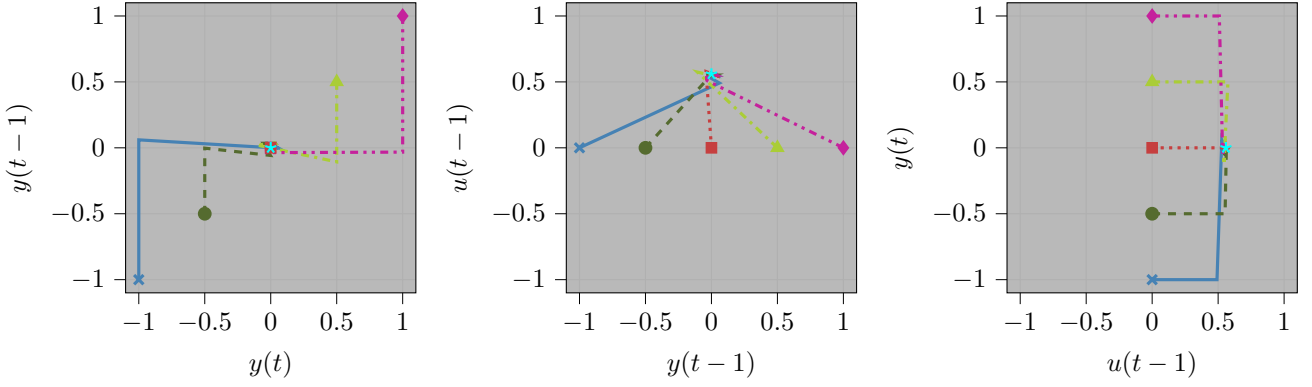
\begin{figure}[t]
    \centering
\begin{tikzpicture}

\definecolor{aqua}{RGB}{0,255,255}
\definecolor{darkgray176}{RGB}{176,176,176}
\definecolor{darkolivegreen}{RGB}{85,107,47}
\definecolor{greenyellow19622665}{RGB}{170,205,55}
\definecolor{indianred1986363}{RGB}{198,63,63}
\definecolor{lightgray204}{RGB}{204,204,204}
\definecolor{mediumvioletred19733156}{RGB}{197,33,156}
\definecolor{silver185}{RGB}{185,185,185}
\definecolor{steelblue}{RGB}{70,130,180}

\begin{groupplot}[group style={group size=3 by 1,horizontal sep=2cm}]
\nextgroupplot[
width=0.3\textwidth,
height=0.3\textwidth,
axis background/.style={fill=silver185},
tick align=outside,
tick pos=left,
x grid style={darkgray176},
xlabel={$y(t)$},
xmajorgrids,
xmin=-1.1, xmax=1.1,
xtick style={color=black},
y grid style={darkgray176},
ylabel={$y(t-1)$},
ymajorgrids,
ymin=-1.1, ymax=1.1,
ytick style={color=black},
xtick={-1,-0.5,0, 0.5, 1},
ytick={-1,-0.5,0, 0.5, 1}
]
\addplot [line width=1.2pt, steelblue,
  mark=x, mark options={solid},
  mark size=2.5pt,mark indices={1}]
table {%
-1 -1
-1 0.0604664450568415
0.0604664450568415 -0.00251995029169416
-0.00251995029169416 0.0340113876004113
0.0340113876004113 0.0340996526333384
0.0340996526333384 0.0341810109649932
0.0341810109649932 0.0341800151467648
0.0341800151467648 0.0341813958501871
0.0341813958501871 0.0341814198847175
0.0341814198847175 0.0341809566669329
};
\addplot [line width=1.2pt,
  mark=*, mark options={solid},
  mark size=2pt, darkolivegreen, mark indices={1}, dashed]
table {%
-0.5 -0.5
-0.5 -0.00247637944630119
-0.00247637944630119 -0.057506052233097
-0.057506052233097 0.0345413311574192
0.0345413311574192 0.034092698772846
0.034092698772846 0.0341829596072012
0.0341829596072012 0.0341816055441346
0.0341816055441346 0.0341795754239111
0.0341795754239111 0.0341816701509483
0.0341816701509483 0.0341824485084432
};
\addplot [line width=1.2pt,
  mark=square*, mark options={solid},
  mark size=1.6pt, indianred1986363,  mark indices={1}, dotted]
table {%
0 0
0 -0.0410664812414052
-0.0410664812414052 0.0347266107472888
0.0347266107472888 0.0340591719538779
0.0340591719538779 0.0341836600398131
0.0341836600398131 0.0341808838218882
0.0341808838218882 0.034180143596044
0.034180143596044 0.0341811875338793
0.0341811875338793 0.0341812839714768
0.0341812839714768 0.0341807650029735
};
\addplot [line width=1.2pt,
  mark=triangle*, mark options={solid},
  mark size=2pt, greenyellow19622665, mark indices={1},dashdotted]
table {%
0.5 0.5
0.5 -0.106530242175488
-0.106530242175488 0.0212823851621247
0.0212823851621247 0.0341614742438514
0.0341614742438514 0.0341371176567784
0.0341371176567784 0.034181220603553
0.034181220603553 0.0341798928364203
0.0341798928364203 0.0341813997052056
0.0341813997052056 0.0341820568958129
0.0341820568958129 0.0341810372091089
};
\addplot [line width=1.2pt,
  mark=diamond*, mark options={solid},
  mark size=2pt, mediumvioletred19733156, mark indices={1}, dashdotdotted]
table {%
1 1
1 -0.0325871093534196
-0.0325871093534196 -0.0366388124527819
-0.0366388124527819 0.033978815131547
0.033978815131547 0.0340843118389325
0.0340843118389325 0.034180208078888
0.034180208078888 0.034180087254402
0.034180087254402 0.0341814113408261
0.0341814113408261 0.0341813553302472
0.0341813553302472 0.0341809548002363
};

\addplot[
  only marks,
  mark=star,
  mark size=2.2pt,
  mark options={draw=aqua,line width=0.8pt}
]
table{
x y
0 0
};

\nextgroupplot[
width=0.3\textwidth,
height=0.3\textwidth,
axis background/.style={fill=silver185},
tick align=outside,
tick pos=left,
x grid style={darkgray176},
xlabel={$y(t-1)$},
xmajorgrids,
xmin=-1.1, xmax=1.1,
xtick style={color=black},
y grid style={darkgray176},
ylabel={$u(t-1)$},
ymajorgrids,
ymin=-1.1, ymax=1.1,
ytick style={color=black},
xtick={-1,-0.5,0, 0.5, 1},
ytick={-1,-0.5,0, 0.5, 1},
]
\addplot [line width=1.2pt, steelblue,
  mark=x, mark options={solid},
  mark size=2.5pt,mark indices={1}]
table {%
-1 0
0.0604664450568415 0.491445894019765
-0.00251995029169416 0.527620043380331
0.0340113876004113 0.552691991811329
0.0340996526333384 0.551824887107642
0.0341810109649932 0.551801007072307
0.0341800151467648 0.551801725468913
0.0341813958501871 0.55180121993056
0.0341814198847175 0.551801233234201
0.0341809566669329 0.551801326359687
};
\addplot [line width=1.2pt,
  mark=*, mark options={solid},
  mark size=2pt, darkolivegreen, mark indices={1},dashed]
table {%
-0.5 0
-0.00247637944630119 0.552765321479725
-0.057506052233097 0.562218822153518
0.0345413311574192 0.551292222631425
0.034092698772846 0.551764555096362
0.0341829596072012 0.551801658950708
0.0341816055441346 0.55180137957425
0.0341795754239111 0.551801605736145
0.0341816701509483 0.551801166715997
0.0341824485084432 0.551801020375947
};
\addplot [line width=1.2pt,
  mark=square*, mark options={solid},
  mark size=1.6pt, indianred1986363,  mark indices={1},dotted]
table {%
0 0
-0.0410664812414052 0.578564367098046
0.0347266107472888 0.550667390531359
0.0340591719538779 0.551850762689106
0.0341836600398131 0.551797867413065
0.0341808838218882 0.551801752076194
0.034180143596044 0.551801472699736
0.0341811875338793 0.551801273145124
0.0341812839714768 0.551801259841483
0.0341807650029735 0.55180136627061
};
\addplot [line width=1.2pt,
  mark=triangle*, mark options={solid},
  mark size=2pt, greenyellow19622665, mark indices={1},dashdotted]
table {%
0.5 0
-0.106530242175488 0.57435336565802
0.0212823851621247 0.544732130807319
0.0341614742438514 0.551746315804749
0.0341371176567784 0.551837312708201
0.034181220603553 0.551800049210165
0.0341798928364203 0.551801698861631
0.0341813997052056 0.55180121993056
0.0341820568958129 0.551801100197792
0.0341810372091089 0.551801313056046
};
\addplot [line width=1.2pt,
  mark=diamond*, mark options={solid},
  mark size=2pt, mediumvioletred19733156, mark indices={1}, dashdotdotted]
table {%
1 0
-0.0325871093534196 0.508980179141126
-0.0366388124527819 0.533870479659882
0.033978815131547 0.552513124360026
0.0340843118389325 0.551788914062769
0.034180208078888 0.551802656723773
0.034180087254402 0.55180168555799
0.0341814113408261 0.55180121993056
0.0341813553302472 0.551801246537842
0.0341809548002363 0.551801326359687
};
\addplot[
  only marks,
  mark=star,
  mark size=2.2pt,
  mark options={draw=aqua,line width=0.8pt}
]
table{
x y
0 0.559
};

\nextgroupplot[
width=0.3\textwidth,
height=0.3\textwidth,
axis background/.style={fill=silver185},
tick align=outside,
tick pos=left,
x grid style={darkgray176},
xlabel={$u(t-1)$},
xmajorgrids,
xmin=-1.1, xmax=1.1,
xtick style={color=black},
y grid style={darkgray176},
ylabel={$y(t)$},
ymajorgrids,
ymin=-1.1, ymax=1.1,
ytick style={color=black},
xtick={-1,-0.5,0, 0.5, 1},
ytick={-1,-0.5,0, 0.5, 1},
]
\addplot [line width=1.2pt, steelblue,
  mark=x, mark options={solid},
  mark size=2.5pt,mark indices={1}]
table {%
0 -1
0.491445894019765 -1
0.527620043380331 0.0604664450568415
0.552691991811329 -0.00251995029169416
0.551824887107642 0.0340113876004113
0.551801007072307 0.0340996526333384
0.551801725468913 0.0341810109649932
0.55180121993056 0.0341800151467648
0.551801233234201 0.0341813958501871
0.551801326359687 0.0341814198847175
};
\addplot [line width=1.2pt,
  mark=*, mark options={solid},
  mark size=2pt, darkolivegreen, mark indices={1},dashed]
table {%
0 -0.5
0.552765321479725 -0.5
0.562218822153518 -0.00247637944630119
0.551292222631425 -0.057506052233097
0.551764555096362 0.0345413311574192
0.551801658950708 0.034092698772846
0.55180137957425 0.0341829596072012
0.551801605736145 0.0341816055441346
0.551801166715997 0.0341795754239111
0.551801020375947 0.0341816701509483
};
\addplot [line width=1.2pt,
  mark=square*, mark options={solid},
  mark size=1.6pt, indianred1986363,  mark indices={1},dotted]
table {%
0 0
0.578564367098046 0
0.550667390531359 -0.0410664812414052
0.551850762689106 0.0347266107472888
0.551797867413065 0.0340591719538779
0.551801752076194 0.0341836600398131
0.551801472699736 0.0341808838218882
0.551801273145124 0.034180143596044
0.551801259841483 0.0341811875338793
0.55180136627061 0.0341812839714768
};
\addplot [line width=1.2pt,
  mark=triangle*, mark options={solid},
  mark size=2pt, greenyellow19622665, mark indices={1},dashdotted]
table {%
0 0.5
0.57435336565802 0.5
0.544732130807319 -0.106530242175488
0.551746315804749 0.0212823851621247
0.551837312708201 0.0341614742438514
0.551800049210165 0.0341371176567784
0.551801698861631 0.034181220603553
0.55180121993056 0.0341798928364203
0.551801100197792 0.0341813997052056
0.551801313056046 0.0341820568958129
};
\addplot [line width=1.2pt,
  mark=diamond*, mark options={solid},
  mark size=2pt, mediumvioletred19733156, mark indices={1}, dashdotdotted]
table {%
0 1
0.508980179141126 1
0.533870479659882 -0.0325871093534196
0.552513124360026 -0.0366388124527819
0.551788914062769 0.033978815131547
0.551802656723773 0.0340843118389325
0.55180168555799 0.034180208078888
0.55180121993056 0.034180087254402
0.551801246537842 0.0341814113408261
0.551801326359687 0.0341813553302472
};
\addplot[
  only marks,
  mark=star,
  mark size=2.2pt,
  mark options={draw=aqua,line width=0.8pt}
]
table{
x y
0.559 0
};
\end{groupplot}

\end{tikzpicture}
    \caption{ Two-dimensional projections of the three-dimensional trajectories $\zeta(t)=[y(t-1);y(t);u(t-1)]$ generated by the proposed controller with initial conditions $[-1;-1;0]$ (blue solid line), $[-0.5;-0.5;0]$ (green dashed line), $[0;0;0]$ (red dotted line), $[0.5;0.5;0]$ (yellow dash-dotted line), and $[1;1;0]$ (violet dash-dot-dotted line). 
    The gray areas correspond to the set defined by \eqref{eq:sumSet}. 
    The initial conditions of each trajectory and the equilibrium point $\zeta^*=[0;0;u^*]$ (cyan star) are indicated by their respective markers.}
    \label{fig:A}
\end{figure}

The output trajectories generated by the proposed controller from five different initial conditions are depicted in Fig.~\ref{fig:traj1}. 
Fig.~\ref{fig:A} shows the corresponding trajectories of the augmented state $\zeta(t)$, projected onto the $(y(t),y(t-1))$, $(y(t-1),u(t-1))$, and $(u(t-1),y(t))$ coordinate planes. 
The gray region represents the set 
\begin{align}\label{eq:sumSet}
    \bigcup_{\delta\in\Delta,~j=1,\ldots,\bar{\kappa}} \calA_\delta^j,
\end{align}
which, in this experiment, covers the entire feasible set of augmented state $\calZ$. 
The cyan star indicates the equilibrium point $\zeta^*=[0;0;u^*]$, where $u^*$ is defined implicitly by
\begin{align*}
    0 = -3 + \sqrt{-(u^*)^2 - 16\ln(u^*)}. 
\end{align*}
Equivalently, $u^*$ satisfies $(u^*)^2 + 16\ln(u^*)+9=0$, which admits a unique solution, numerically given by $u^*\approx 0.5588$.
That is, the output remains zero when the input is held constant at $u(t)=u^*$.
It is observed that the trajectories of the augmented state converge toward the equilibrium point $\zeta^*$ for all considered initial conditions, indicating that the proposed controller achieves practical output regulation and validating the guarantees established in Theorem~\ref{thm:main}. 

In practice, it is recommended to select $\Delta$ and $\bar{\kappa}$ sufficiently large, which enlarges the set defined by \eqref{eq:sumSet}. 
This increases the likelihood that the condition $\zeta(0) \in \calA_\delta^\kappa$ in \eqref{eq:suff} is satisfied for some $\delta\in\Delta$ and $0\le \kappa \le \bar{\kappa}$.

\subsection{Inverted pendulum case study}

Next, we consider the standard inverted pendulum as shown in Fig.~\ref{fig:inv_pendulum}. 
Here, $m\in\bbR$ and $l\in\bbR$ are the mass and length of the pendulum, respectively, $\theta\in\bbR$ is the angular displacement measured from the upright equilibrium, and $\tau\in\bbR$ is the applied torque. 
The equation of motion can be directly derived as
\begin{align}\label{eq:pendulumCont}
    ml^2 \ddot{\theta} + b\dot{\theta} - mgl \sin(\theta) = \tau,
\end{align}
where $g\in\bbR$ denotes the gravitational acceleration and $b\in\bbR$ denotes the coefficient of the viscous friction.

To obtain a discrete-time representation with a sampling period of $T_s>0$, we employ the standard discretization 
\begin{align*}
    \ddot{\theta}(tT_s) &\approx \frac{\theta(tT_s + 2T_s) - 2\theta(tT_s + T_s) + \theta(tT_s)}{T_s^2}, \qquad
    \dot{\theta}(tT_s) \approx \frac{\theta(tT_s + T_s)-\theta(tT_s)}{T_s}
\end{align*}
for $t\in\bbZ$, and define the discrete-time input and output as $u(t) = \tau(tT_s)$ and $y(t) = \theta(tT_s)$, respectively.
Substituting these expressions into \eqref{eq:pendulumCont} yields a discrete-time NARX model of the form \eqref{eq:sysDelay} with an input delay $\nu=2$:
\begin{align} \label{eq:pendulumDisc}
    y(t+2) &= \left(2-\frac{bT_s}{ml^2} \right) y(t+1)
    + \left(-1 + \frac{bT_s}{ml^2} \right)y(t)
    + \frac{gT_s^2}{l} \sin(y(t))
    + \frac{T_s^2}{ml^2} u(t).
\end{align}
After some algebraic manipulation, the corresponding inverse model of \eqref{eq:pendulumDisc} can be obtained as 
\begin{align*}
    u(t) &= c([y(t+2);\zeta(t)]) \\
    &= \frac{ml^2}{T_s^2} \bigg(
    y(t+2)
    - \left(3 - \frac{3bT_s}{ml^2} + \frac{b^2T_s^2}{m^2l^4} \right) y(t) - \frac{gT_s^2}{l}  \sin(y(t)) \\
    &~~~~ - \left(2 - \frac{bT_s}{ml^2} \right)  \left( 
    \left(-1 + \frac{bT_s}{ml^2} \right)y(t-1)
    + \frac{gT_s^2}{l} \sin(y(t-1))
    + \frac{T_s^2}{ml^2} u(t-1) \right)
            \bigg) .
\end{align*}
In this experiment, the parameter values are set to 
\begin{align*}
    m = \SI{1}{\kilogram}, \qquad
    b = \SI{0.4}{\newton\second/\meter}, \qquad
    g = \SI{9.8}{\meter/\second\squared}, \qquad
    l = \SI{0.3}{\meter}, \qquad
    T_s = \SI{0.001}{\second}.
\end{align*}

\subsubsection{Noise-free case}

\begin{figure}[t]
\centering
\begin{tikzpicture}[
  scale=1.0,
  line cap=round, line join=round,
  >={Latex[length=2.2mm]}
]
  \def\L{3.2}            
  \def\th{22}            
  \def\massR{0.16}       
  \def\pivotR{0.06}      
  \def\thetR{2}       

  \coordinate (O) at (0,0);                               
  \coordinate (Up) at (0,\L);                             
  \coordinate (P)  at ({\L*sin(\th)},{\L*cos(\th)});      

  \draw[thick] (-1.2,0) -- (1.2,0);

  \fill (O) circle (\pivotR);

  \draw[densely dashed] (O) -- (Up);

  \draw[thick] (O) -- (P);
  \path (O) -- node[midway, right=2pt] {$l$} (P);

  \fill[white] (P) circle (\massR);
  \draw[thick] (P) circle (\massR);
  \node[right=6pt] at (P) {$m$};

\draw[thick, {Latex[length=2mm,width=1.4mm]}-{Latex[length=2mm,width=1.4mm]}]
  (O) ++(0,\thetR)
  arc[start angle=90, end angle={90-\th}, radius=\thetR];

\pgfmathsetmacro{\thmid}{90 - 0.5*\th}
\node at ($(O)+({0.90*\thetR*cos(\thmid)},{0.90*\thetR*sin(\thmid)})$) {$\theta$};

\def\rtau{0.45}
\def\tauA{15}   
\def\tauB{135} 
\draw[thick,->]
  ($(O)+({\rtau*cos(\tauA)},{\rtau*sin(\tauA)})$)
  arc[start angle=\tauA, end angle=\tauB, radius=\rtau];

\node[right=3pt] at ($(O)+({(\rtau+0.10)*cos(65)},{(\rtau+0.10)*sin(65)})$) {$\tau$};

\end{tikzpicture}
\caption{Inverted pendulum with torque input \(\tau\) and parameters $m$ and $l$.}
\label{fig:inv_pendulum}
\end{figure}
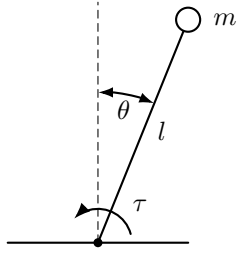

The training dataset $\calD$ was constructed by combining six closed-loop trajectories, each of length $200$, generated by three proportional-integral (PI) controllers \cite{Ogat10} with different gains written in the form
\begin{align*}
    u(t) &= -K_p e(t) - K_I I(t), \\
    e(t) &= y(t), \\
    I(t) &= I(t-1) + T_s e(t), \qquad I(0)=0,
\end{align*}
where $K_p\ge 0$ and $K_I\ge 0$ denote the proportional and integral gains, respectively.
For each controller, we generated two trajectories from distinct initial conditions of the form $\zeta(0)=[a;\,a;\,0]$.
Specifically, the triplets $(K_p,K_I,a)$ used in the experiment were $(20,0.01,\pm0.22)$, $(15,0.01,\pm0.18)$, and $(12.5,0.01,\pm0.16)$.
In summary, a total number of $N=1200$ training data were used. 
This setup reflects an expert-mimicking scenario, in which data generated by an (unknown) expert controller are available while the controller structure and/or parameters are unknown.

We utilized the automatic relevance determination (ARD) Mat\'ern-$5/2$ kernel \cite{WillRasm06}, written by
\begin{align*}
    k(\xi,\xi') = \sigma_f^2 \left( 1+\sqrt{5}r + \frac{5}{3}r^2 \right) \exp(-\sqrt{5}r), ~~~~ r=\sqrt{\sum_{i=1}^4 \frac{(\xi_i-\xi_i')^2}{2\sigma_{l,i}^2} }
\end{align*}
where $\xi_i$ and $\xi_i'$ denote the $i$-th element of $\xi\in\bbR^4$ and $\xi'\in\bbR^4$, respectively.
While alternative kernels that better capture the geometric properties of the inverse model may improve performance \cite[Section~2.3]{Duve14}, we focus on assessing the proposed method using the standard Mat\'ern kernel as a baseline choice.
Instead of KI, we obtained an estimate of the inverse model using the sparse variational Gaussian process (SVGP) \cite{Tits09,GardPlei18} with $192$ inducing points, as it improves computational scalability and naturally accommodates measurement noise.
The kernel hyperparameters $\sigma_f$ and $\{\sigma_{l,i}\}_{i=1}^4$ were optimized by maximizing the variational evidence lower bound (ELBO) of $\mathcal{D}$, using the Adam optimizer with a learning rate of $0.01$ over $150$ iterations.
For the implementation of Algorithm~\ref{alg:pgigp}, we set $\gamma(\epsilon) = 1.005 \epsilon$ and $\bar{\kappa} = 100$, and chose $\Delta = \{0.01,0.02,0.03,0.04,0.05,0.06,0.08,0.1,0.3,0.6\}$. 

The training trajectories obtained from the PI controllers are depicted in the top-left subplot of Fig.~\ref{fig:traj2}.
The middle-left and bottom-left subplots of Fig.~\ref{fig:traj2} show the output trajectories generated by the proposed controller and by a baseline PI controller\footnote{While three PI controllers were used to generate the training dataset, we report results for a single representative PI controller to maintain figure clarity.} with gains $(K_p,K_I)=(15,0.01)$, respectively.
All trajectories are simulated for $T_{\Sim}=500$ time steps from four initial conditions $[0.1;0.1;0]$, $[-0.1;-0.1;0]$, $[0.05;0.05;0]$, and $[-0.05;-0.05;0]$.
As a performance metric, we consider the root mean square error (RMSE), defined by 
\begin{align*}
    \RMSE(y_{[0,T_\Sim]}) := \frac{1}{T_\Sim +1}\sqrt{\sum_{t=0}^{T_\Sim} \left|y(t) \right|^2},
\end{align*}
which quantifies the overall regulation error.
The results are reported in Table~\ref{tab:performance}.
It can be seen that the training trajectories are oscillatory and do not fully regulate the output. Nonetheless, the proposed controller achieves practical output regulation from all considered initial conditions, with  performance comparable to that of the baseline PI controller in terms of RMSE.


\begin{figure}[t]
    \centering
        \input{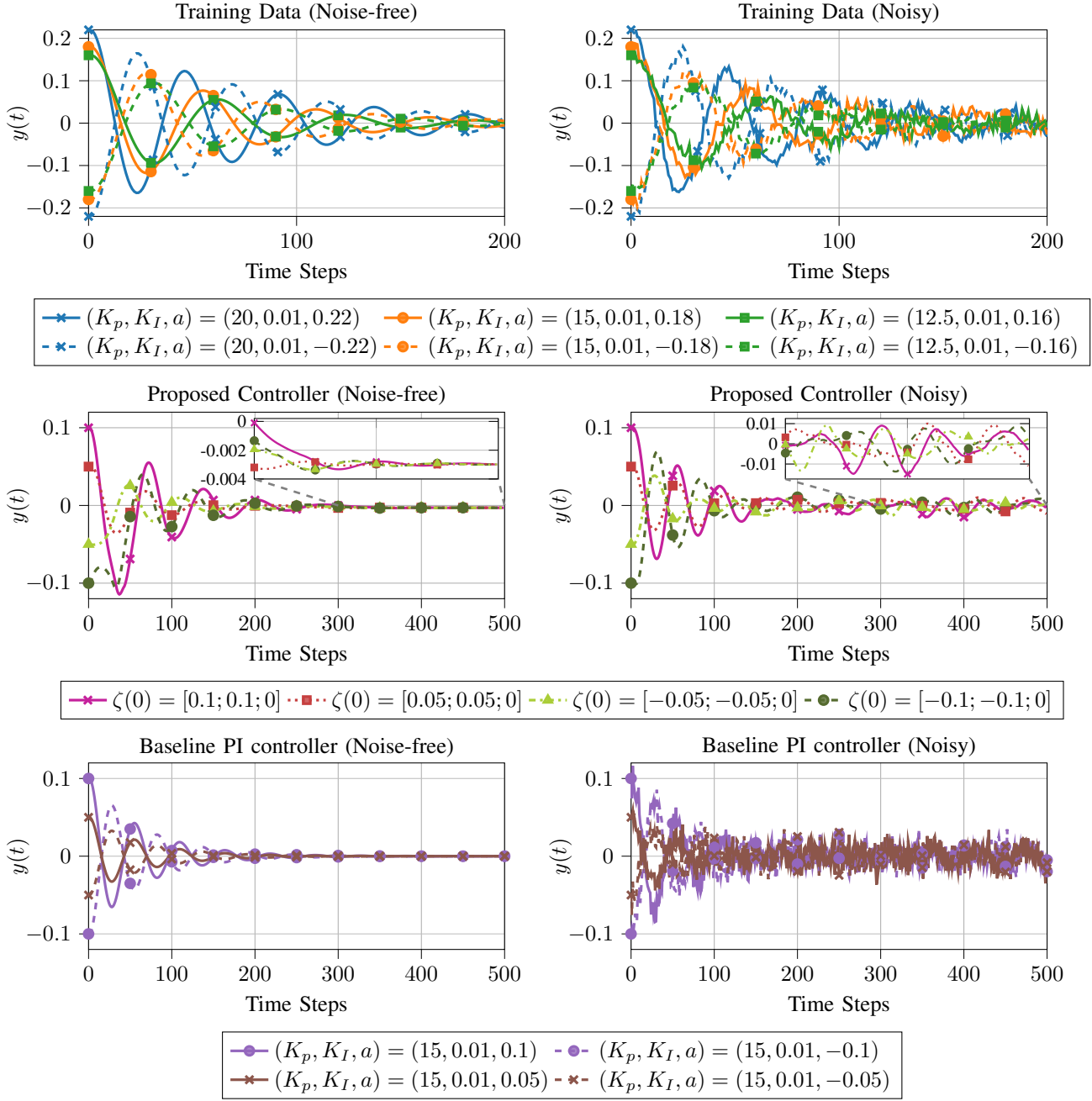}
    \caption{Noise-free (top-left) and noisy (top-right) training trajectories generated by the PI controllers with different gains $(K_p,K_I)$ and initial condition $\zeta(0)=[a;a;0]$. 
    The middle-left and middle-right subplots show the output trajectories from the proposed controller in the noise-free and noisy cases, respectively, from several initial conditions.
    The bottom-left and bottom-right subplots show the corresponding noise-free and noisy output trajectories generated by the baseline PI controller with gains $(K_p,K_I)=(15,0.01)$, respectively.}
    \label{fig:traj2}
\end{figure}

\begin{table*}[htbp]
\centering
\caption{ RMSE comparison of the proposed controller and the baseline PI controller under noise-free and noisy measurements.}
\label{tab:performance}
\begin{tabular}{|c|c|c|c|c|}
\hline
\multirow{2}{*}{\textbf{Initial condition} $\zeta(0)$} & \multicolumn{2}{c|}{\textbf{RMSE (Noise-free)}} & \multicolumn{2}{c|}{\textbf{RMSE (Noisy)}} \\
\cline{2-5}
 & \textbf{Proposed} & \textbf{Baseline PI} & \textbf{Proposed} & \textbf{Baseline PI} \\
\hline
[0.1;0.1;0] & 0.0212 & 0.0192  
            & 0.0162 & 0.0218 \\ 
\hline 
[-0.1;-0.1;0] & 0.0203 & 0.0192  
              & 0.0168 & 0.0220  \\ 
\hline
[0.05;0.05;0] & 0.0081 & 0.0096  
              & 0.0086 & 0.1462  \\ 
\hline
[-0.05;-0.05;0] & 0.0093 & 0.0096 
                & 0.0090 & 0.0142 \\ 
\hline
\end{tabular}
\end{table*}

\subsubsection{Noisy case}
We now consider the case in which the output measurements of \eqref{eq:sysDelay} are corrupted by additive noise, in order to empirically assess the robustness of the proposed controller.
Specifically, we assume that only noisy measurements 
\begin{align*}
    \tilde{y}^\rmd(t) := y^\rmd(t) + v^\rmd (t), ~~~~ t=0,\ldots, T, 
\end{align*}
are available for constructing the dataset $\calD$, where $v^\rmd(t)\sim \calN(0,(\sigma^\rmd)^2)$ denotes a zero-mean Gaussian noise with a standard deviation $\sigma^\rmd \ge 0 $. 
During the online procedure, the output measurements are similarly assumed to be corrupted as 
\begin{align*}
    \tilde{y}(t) := y(t) + v(t), ~~~ t\in \bbZ,
\end{align*}
where $v(t)\sim \calN(0,\sigma^2)$ with $\sigma \ge 0$. 
The controller \eqref{eq:igp2} is accordingly modified as
\begin{align*}
    u(t) = \hat{c} ([y_\rmr(t+2); \tilde{\zeta}(t)]),
\end{align*}
where $\tilde{\zeta}(t):=[\tilde{y}_{[t-n+1,t]}; u_{[t-n+1,t-1]}]$.

We evaluated the proposed controller and the baseline PI controller with the noise standard deviations set to $(\sigma^\rmd, \sigma)= (0.01, 0.01)$.
The top-right subplot of Fig.~\ref{fig:traj2} shows the noisy training trajectories, and the middle-right and bottom-right subplots show the resulting output trajectories generated by the proposed controller and the baseline PI controller, respectively.

It can be seen that the proposed controller still achieves practical output regulation and attains smaller RMSE than the baseline PI controller, also exhibiting reduced oscillations and less chattering.
Compared to the noise-free case, however, a larger steady-state offset is observed, which can be attributed to measurement noise. 
Overall, these results suggest that the proposed controller remains effective under noisy measurements, motivating further work on extending the framework to explicitly account for noise and provide formal guarantees in such settings.


\section{Conclusion}\label{sec:conclusion}
We have presented a data-driven output feedback controller for systems represented in NARX form, using input/output measurement data. 
The controller is constructed by identifying an inverse model of the system via KI and by combining it with a data-driven reference selection framework, described in Section~\ref{sec:control}. 
By leveraging the interpolation error bound of KI, we have established a verifiable sufficient condition on the dataset under which the proposed controller guarantees practical output regulation. 
We have validated the practical utility of the proposed method through numerical simulations.
Future work will consider an explicit treatment of measurement noise arising from sensor degradation and extensions to vector-valued kernel methods to capture coupling effects across multiple inputs.

\bibliographystyle{IEEEtran}
\bibliography{v1/wileyNJD-AMA}

\end{document}